\documentclass[12pt]{article}
\usepackage{graphicx} 

\usepackage{style}
\usepackage{authblk}

\title{Generalized data thinning
using sufficient statistics}
\author[1]{Ameer Dharamshi}
\author[2]{Anna Neufeld}
\author[1]{Keshav Motwani}
\author[3]{Lucy L. Gao}
\author[1,4]{Daniela Witten}
\author[5]{Jacob Bien}
\affil[1]{Department of Biostatistics, University of Washington}
\affil[2]{Public Health Sciences Division, Fred Hutchinson Cancer Research Center}
\affil[3]{Department of Statistics, University of British Columbia}
\affil[4]{Department of Statistics, University of Washington}
\affil[5]{Department of Data Sciences and Operations, University of Southern California}

\begin{document}

\maketitle


\begin{abstract}
Our goal is to develop a general strategy to decompose a random variable $X$ into multiple independent random variables, without sacrificing any information about unknown parameters.
A recent paper showed that for some well-known natural exponential families, $X$ can be \emph{thinned} into independent random variables $\xo, \ldots, \Xt{K}$, such that $X = \sum_{k=1}^K \Xt{k}$. These independent random variables can then be used for various model validation and inference tasks, including in contexts where traditional sample splitting fails. In this paper, we generalize that procedure by relaxing the summation requirement and simply asking that some known function of the independent random variables exactly reconstruct $X$. This generalization of the procedure serves two purposes.  First, it greatly expands the families of distributions for which thinning can be performed.  Second, it unifies sample splitting and data thinning, which on the surface seem to be very different, as applications of the same principle.  This shared principle is sufficiency. We use this insight to perform generalized thinning operations for a diverse set of families.
\end{abstract}

\section{Introduction}
\label{sec:introduction}

Suppose that we want to \emph{fit} and \emph{validate} a model using 
a single dataset.  Two example scenarios are as follows:
\begin{list}{}{}
\item{\emph{Scenario 1.}} We  want to use the data both to generate and to test a hypothesis. 
\item{\emph{Scenario 2.}} We want to use the  data both to fit a complicated model, and to obtain an accurate estimate of the expected prediction error. 
\end{list}
In either case, a naive approach that fits and validates a model on the same data is deeply problematic. In Scenario 1, testing a hypothesis on the same data used to generate it will lead to hypothesis tests that do not control the type 1 error, and to confidence intervals that do not attain the nominal coverage \citep{fithian2014optimal}.   And in Scenario 2, estimating the expected prediction error on the same data used to fit the model will lead to massive downward bias  \citep[see][for recent reviews]{tian2020prediction,oliveira2021unbiased}.

In the case of Scenario 1, recent interest  has focused on \emph{selective inference}, a framework that enables a data analyst to generate and test a hypothesis on the same data \citep[see, e.g.,][]{taylor2015statistical}. The main idea is as follows: to test a hypothesis generated from the data, we should condition on the event that we selected this particular hypothesis. Despite promising applications of this framework to a number of problems, such as inference after regression \citep{lee2016exact}, changepoint detection \citep{jewell2022testing,hyun2021post}, clustering \citep{gao2020selective,chen2022selective,yun2023selective}, and outlier detection \citep{chen2020valid}, it suffers from some drawbacks: 
\begin{enumerate}
\item To perform selective inference, the  procedure used to generate the null hypothesis must be fully-specified in advance.  For instance, if a researcher wishes to cluster the data and then test for a difference in means between the clusters, as in \cite{gao2020selective} and \cite{chen2022selective}, then they must fully specify the clustering procedure (e.g., hierarchical clustering with squared Euclidean distance and complete linkage, cut to obtain $K$ clusters) in advance. 
\item Finite-sample selective inference typically requires multivariate Gaussianity, though in some cases this can be relaxed to obtain asymptotic results \citep{taylor2018post,tian2017asymptotics,tibshirani2018uniform,tian2018selective}.
\end{enumerate}
Thus, selective inference is not a flexible, ``one-size-fits-all" approach to Scenario 1. 

In the case of Scenario 2, proposals to de-bias the ``in-sample" estimate of expected prediction error  tend to be specialized to simple models, and thus do not provide an all-purpose tool that is broadly applicable \citep{oliveira2021unbiased}.

\emph{Sample splitting} \citep{cox1975note} is an intuitive  approach that applies to a variety of settings, including Scenarios 1 and 2; see the left-hand panel of Figure~\ref{fig:samplesplit_vs_datathin}. We split a dataset containing $n$ observations into two sets, containing $n_1$ and $n_2$ observations (where $n_1+n_2=n$). Then we can generate a hypothesis based on the first set and test it on the second (Scenario 1), or we can fit a model to the first set and estimate its error on the second (Scenario 2). Sample splitting also forms the basis for cross-validation \citep{hastie2009elements}. 

However, sample splitting 
suffers from some drawbacks: 
\begin{enumerate} 
    \item If the data contain outliers, then each outlier is assigned to a single subsample. 
    \item If the observations are not independent (for instance, if they correspond to a time series) then the subsamples from sample splitting are not independent, and so sample splitting does not provide a solution to either Scenario 1 or Scenario 2.
    \item Sample splitting does not enable conclusions at a per-observation level.  
    For example, if sample splitting is applied to a dataset of the 50 states of the United States, then one can only conduct inference or perform validation on the states not used in fitting.
    \item If the model of interest is fit using  unsupervised learning, then  sample splitting may  not provide an adequate solution in either Scenario 1 or 2.  The issue relates to \#3 above. See  \cite{gao2020selective,chen2022selective}, and \cite{neufeld2022inference}. 
\end{enumerate}

In recent work, \cite{neufeld2023data} proposed \emph{convolution-closed data thinning} to address these drawbacks. They consider splitting, or \emph{thinning}, a random variable $X$ drawn from a convolution-closed family into $K$ independent random variables $\Xt{1},\ldots,\Xt{K}$ such that
$X=\sum_{k=1}^K\Xt{k}$, 
and $\Xt{1},\ldots,\Xt{K}$ come from the same family of distributions as $X$ (see the right-hand panel of Figure~\ref{fig:samplesplit_vs_datathin}). 
For instance, they show that $X \sim N(\mu, \sigma^2)$ can be thinned into two independent $N(\epsilon \mu, \epsilon \sigma^2)$ and $N((1-\epsilon) \mu, (1-\epsilon) \sigma^2)$ random variables that sum to $X$. 
Further, if  $X$ is drawn from a Gaussian, Poisson, negative binomial, binomial, multinomial, or gamma distribution, then they can thin it  \emph{even when parameters of its distribution are unknown}. Because the thinned random variables are independent, this
provides a new approach to tackle Scenarios 1 and 2: After thinning the data into independent parts, we fit a model to one part, and validate it on the rest.

On the surface, it is quite remarkable that one can break up a random variable $X$ into two or more {\em independent} random variables that sum to $X$ without knowing some (or sometimes any) of the parameters.  In this paper, we explain the underlying principles that make this possible. We also show that convolution-closed data thinning can be generalized to increase its flexibility and applicability. The convolution-closed data thinning property $X=\sum_{k=1}^K\Xt{k}$ is desirable because it ensures that no information has been lost in the thinning process. However, clearly this would remain true if we were to replace the summation by any other deterministic function.  Likewise, the fact that $\Xt{1},\ldots,\Xt{K}$ are from the same family as $X$, while convenient, is nonessential. 

Our generalization of convolution-closed data thinning is thus a procedure for splitting $X$ into $K$ random variables such that  the following two properties hold: 

\begin{itemize}
  \item[] (i) $X=T(\Xt{1},\ldots,\Xt{K})$; and (ii) $\Xt{1},\ldots,\Xt{K}$ are mutually independent.
\end{itemize}

This generalization is broad enough to simultaneously encompass both convolution-closed data thinning and sample splitting. Furthermore, it greatly increases the scope of distributions that can be thinned. In the $K=2$ case, this generalized goal has been stated before \citep[see][``P1'' property]{leiner2022data}. However, we are the first to develop a widely applicable strategy for achieving this goal. Not only can we thin exponential families that were not previously possible (such as the  beta family), but we can even thin outside of the exponential family.  For example, generalized thinning enables us to thin $X \sim \text{Unif}(0, \theta)$ into
 $\Xt{k} \overset{\text{iid}}{\sim} \theta \cdot\text{Beta}\left(\frac{1}{K},1\right)$, for $k=1,\dots,K$, in such a way that $X=\max\{\Xt{1},\ldots,\Xt{K}\}$.

The primary contributions of our paper are as follows:
\begin{enumerate}
\item We propose \emph{generalized data thinning}, a general strategy for thinning a single random variable $X$ into two or more independent random variables, $\xo,\ldots,\Xt{K}$, without knowledge of the parameter value(s).  
Importantly, we show that {\em sufficiency} is the key property underlying the choice of the function $T(\cdot)$.
\item We apply generalized data thinning to distributions far outside the scope of consideration of \cite{neufeld2023data}: These include the beta, uniform, and shifted exponential, among others.  A summary of distributions covered by this work is provided in Table~\ref{table:maintable}.  In light of results by \cite{darmois, koopman}, and \cite{pitman_1936}, we believe our examples are representative of the full range of cases to which this approach can be applied. 
\item We show that sample splitting --- which, on its surface, bears little resemblance to convolution-closed data thinning --- is in fact based on the same principle: Both are special cases of generalized data thinning with different choices of the function $T(\cdot)$.  In other words, our proposal is a direct \emph{generalization} of sample splitting. 
\end{enumerate}

We are not the first to generalize sample splitting.  Inspired by \cite{tian2018selective}'s use of randomized responses, \cite{rasines2021splitting} introduce the ``$(U,V)$-decomposition", which injects independent noise $W$ to create two independent random variables $U=u(X,W)$ and $V=v(X,W)$ that together are jointly sufficient for the unknown parameters.  However, they do not describe how to perform a $(U,V)$-decomposition other than in the special case of a Gaussian random vector with known covariance.  Our generalized thinning framework achieves the goal set out in their paper, providing a concrete recipe for finding such decompositions in a broad set of examples.  The ``data fission" proposal of \cite{leiner2022data} seeks random variables $f(X)$ and $g(X)$ for which the distributions of $f(X)$ and $g(X) \mid f(X)$ are known and for which $X=h(f(X),g(X))$. When these two random variables are independent (the ``P1'' property), their proposal aligns with generalized thinning.  However, they do not provide a general strategy for performing P1-fission, and the only two examples they provide are the Gaussian vector with known covariance and the Poisson.

The rest of our paper is organized as follows. In Section~\ref{sec:method}, we define generalized data thinning, present our main theorem, and provide a simple recipe for thinning that is followed throughout the paper. Sections~\ref{sec:natural-exp-fam}--\ref{sec:outside-exp-fam} demonstrate the utility of our approach in a series of examples organized by the results of \cite{darmois, koopman}, and \cite{pitman_1936}: In particular, in Section~\ref{sec:natural-exp-fam}, we consider the case of thinning natural exponential families; this section also revisits the convolution-closed data thinning proposal of \cite{neufeld2023data} and clarifies the class of distributions that can be thinned using that approach. In Section~\ref{sec:general-exp}, we apply data thinning to  general exponential families. We consider distributions outside of the exponential family in Section~\ref{sec:outside-exp-fam}. Section~\ref{sec:counterexamples} contains examples of distributions that \emph{cannot} be thinned using the approaches in this paper.  
Section~\ref{sec:changepoint} presents an application of data thinning to changepoint detection.  
Finally, we close with a discussion in Section~\ref{sec:discussion}; additional technical details are deferred to the supplementary materials.  

\begin{figure}
  \hspace{39mm}  Sample splitting    \hspace{10mm} Generalized data thinning 
  \vspace{-3mm}
\begin{center} 
\centering
\includegraphics[scale=0.18,trim={3cm 8cm 39cm 8cm},clip]{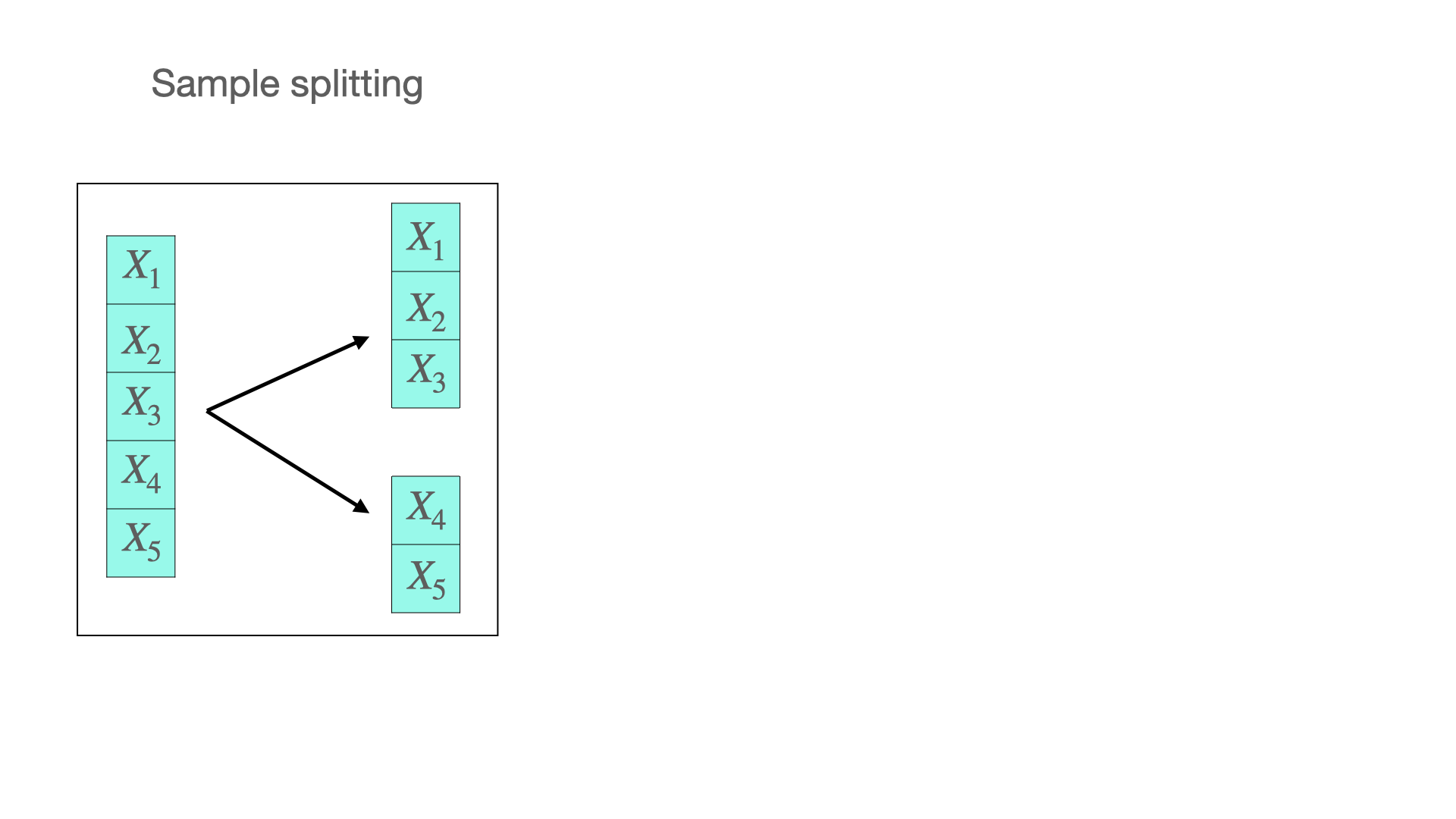}
\includegraphics[scale=0.18,trim={3cm 8cm 40cm 8cm},clip]{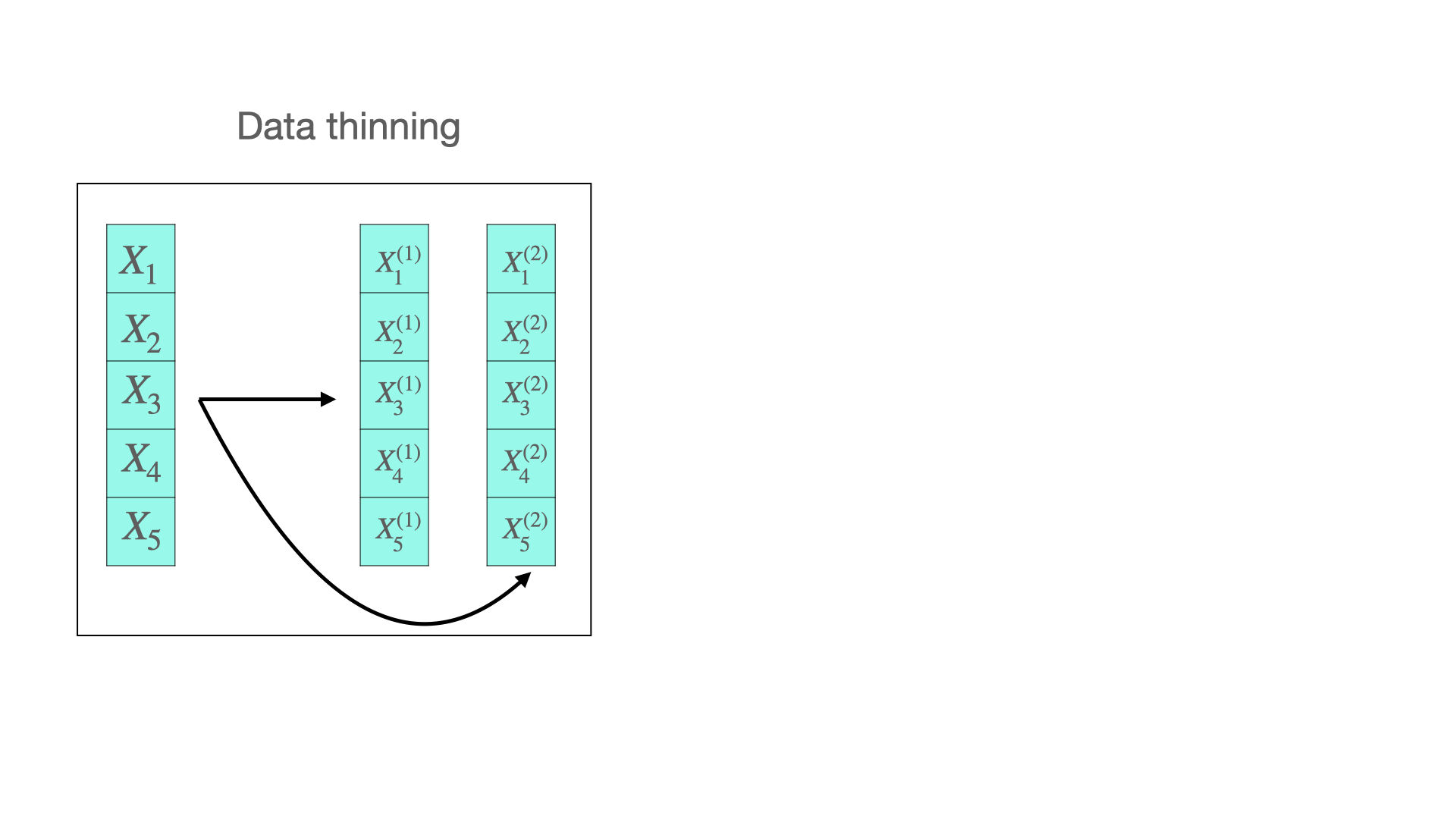} 
\caption{\emph{Left:} Sample splitting assigns each observation to either a training or a test set. 
\emph{Right:} Generalized data thinning splits each observation into two parts that are independent and  can be used to recover the original observation, i.e. $T(\xo, \Xt{2})=X$. 
\label{fig:samplesplit_vs_datathin}}
\end{center}
\end{figure}

\begin{table}
\scriptsize
\begin{spacing}{1.5}
\begin{tabular}{|c || c | c ||c | c |}
\hline
Family & Distribution $P_{\theta}$, & Distribution $Q_{\theta}^{(k)}$ & Sufficient statistic $T$ & Reference / notes \\
& where $X \sim P_{\theta}$. & where $X^{(k)} \overset{ind.}{\sim} Q_{\theta}^{(k)}$. & (sufficient for $\theta$) & \\
\hline
\hline
\multirow{9}{*}{
\begin{tabular}{c}
Natural \\exponential \\ family \\
(in parameter $\theta$)
\end{tabular}
}
& $N(\theta, \sigma^2)$ & $\mathrm{N}(\epsilon_k \theta, \epsilon_k \sigma^2)$ & \multirow{5}{*}{$\sum_{k=1}^K X^{(k)}$}  & \multirow{7}{*}{\cite{neufeld2023data}} \\
& $\mathrm{Poisson}(\theta)$ & $\mathrm{Poisson}(\epsilon_k \theta)$ &   &  \\
& $\mathrm{NegBin}(r, \theta)$ & $\mathrm{NegBin}(\epsilon_k r, \theta)$ & &   \\
& $\mathrm{Binomial}(r,\theta) $ & $\mathrm{Binomial}(\epsilon_k r,\theta) $& & \\
& $\mathrm{Gamma}(\alpha, \theta)$ & $\mathrm{Gamma}(\epsilon_k \alpha, \theta)$ & & \\
\cline{2-4}
& $\mathrm{N}_p(\boldsymbol\theta
, \boldsymbol\Sigma)$ & $N_p(\eps_k \boldsymbol\theta
, \eps_k \boldsymbol\Sigma)$ & \multirow{2}{*}{$\sum_{k=1}^K \mathbf{X}^{(k)}$} &   \\
& $\mathrm{Multinomial}_p(r, \boldsymbol\theta
)$ & 
$\mathrm{Multinomial}_p(\eps_k r, \boldsymbol\theta
)$  & & \\
\cline{2-5}
&  $\mathrm{Gamma}(K/2, \theta)$ &  $N(0, \frac{1}{2 \theta})$ & $\sum_{k=1}^K \left(X^{(k)}\right)^2$ &  Example~\ref{ex:scaled-normal} \\
&  $\mathrm{Gamma}(K, \theta)$ &  $\mathrm{Weibull}(\theta^{-\frac{1}{\nu}}, \nu)$ &  $\sum_{k=1}^K \left(X^{(k)}\right)^\nu$ & 
 Example C.1  \\
\hline 
\multirow{6}{*}{
\begin{tabular}{c}
General \\exponential \\
family \\
(in parameter $\theta$)
\end{tabular}
} 

& $\mathrm{Beta}(\theta,\beta)$ & 
$\mathrm{Beta}\left(\frac{1}{K}\theta+\frac{k-1}{K}, \frac{1}{K}\beta\right)$ & $\left( \Pi_{k=1}^K X^{(k)}\right)^{1/K}$ &  Example~\ref{ex:beta} \\
& $\mathrm{Beta}(\alpha, \theta)$ & $\mathrm{Beta}\left(\frac{1}{K}\alpha, \frac{1}{K}\theta+\frac{k-1}{K}\right)$ & $\left( \Pi_{k=1}^K \left(1-X^{(k)}\right)\right)^{1/K}$ & {\tiny Text below Example~\ref{ex:beta}} \\
& $\mathrm{Gamma}(\theta, \beta)$ & $\mathrm{Gamma}(\frac{1}{K} \theta + \frac{k-1}{K} , \frac{1}{K} \beta)$ & $\left( \Pi_{k=1}^K X^{(k)}\right)^{1/K}$ &
  Example~\ref{ex:gamma1} \\
 & $\mathrm{Weibull}(\theta, \gamma)$ &  $\mathrm{Gamma}(\frac{1}{K}, \theta^{-\gamma})$ & $\left( \sum_{k=1}^K X^{(k)} \right)^{1/\gamma}$ &  Example~\ref{ex:other}  \\
 & $\mathrm{Pareto}(\gamma, \theta)$ & $\mathrm{Gamma}(\frac{1}{K}, \theta)$  & $\gamma \times \mathrm{Exp} \left( \sum_{k=1}^K X^{(k)} \right)$ &  Example~\ref{ex:other}  \\
 & $\mathrm{Dirichlet}_K(\boldsymbol\theta, \phi)$ & $\mathrm{Gamma}(\theta_k\phi, \nu)$  & $\left(\Xt{1},\dots,\Xt{K}  \right)^\top/\sum_{k=1}^K X^{(k)}$ &  Example C.2  \\
 \cline{2-5}
& $\mathrm{N}(\mu, \theta)$ & $\mathrm{Gamma}(\frac{1}{2K}, \frac{1}{2 \theta})$ & $(X-\mu)^2 = \sum_{k=1}^K X^{(k)}$ &{\tiny Indirect only;  Example~\ref{ex:other} }\\
& $\mathrm{N}_K(\theta_1 \mathbf{1}_K, \theta_2\mathbf{I}_K)$ & $N(\theta_1,\theta_2)$ & sample mean and variance &{\tiny Indirect only;  Example D.1}\\
\hline 
\multirow{3}{*}{
\begin{tabular}{c}
Truncated \\support \\
family
\end{tabular}
}
& $\mathrm{Unif}(0,\theta)$ & $\theta \cdot\mathrm{Beta}(\frac{1}{K}, 1)$ & \multirow{2}{*}{$\mathrm{max} \left( X^{(1)}, \ldots, X^{(K)}\right)$} &  Example~\ref{ex:scaled-uniform} \\ 
& $\theta \cdot\mathrm{Beta}(\alpha, 1)$ & $\theta \cdot\mathrm{Beta}(\frac{\alpha}{K}, 1)$ & & Example C.3 \\
\cline{2-5}
& $\theta + \mathrm{Exp}(\lambda)$ & $\theta + \mathrm{Exp(\lambda/K)}$ & $\mathrm{min}\left(X^{(1)},\ldots, X^{(K)} \right)$ &  Example C.4\\
\hline 
Non-parametric &
$F^n$  & $F^{n_k}$ & 
 See Example~\ref{ex:samplesplit} 
 & Example~\ref{ex:samplesplit} \\
\hline
\end{tabular}
\end{spacing}
\caption{
Examples of named families (indexed by an unknown parameter $\theta$) that can be thinned into $K$ components, where $K$ is a positive integer, without knowledge of $\theta$.  
In cases where they are used, $\epsilon_k$, $n_k$, and $\nu$ are positive tuning parameters to be selected by the analyst, where $\sum_{k=1}^K \epsilon_k = 1$ and $n_1,\dots,n_K$ are integers that sum to $n$; all other parameters are constrained appropriately. Note that Examples C.1, C.2, C.3, C.4, and D.1 are discussed in the supplementary materials.
}
\label{table:maintable}
\end{table}

\section{The generalized thinning proposal}
\label{sec:method}
We write $X$ to denote a random variable that can be scalar-, vector-, or matrix-valued (and likewise for $\Xt{1},\ldots,\Xt{K}$). When referring to a random variable or parameter that can only be vector- or matrix-valued, we use bolded symbols.
\begin{definition}[Generalized  data thinning] 
\label{def:thinning}
Consider a family of distributions $\mathcal P=\{P_\theta:\theta\in\Omega\}$. Suppose that there exists a distribution $G_t$, not depending on $\theta$, and a deterministic function $T(\cdot)$ such that when we sample $(\Xt{1},\ldots,\Xt{K})|X$ from $G_X$, for  $X \sim P_\theta$, the following properties hold: 
\begin{enumerate}
  \item $\Xt{1},\ldots,\Xt{K}$ are mutually independent (with distributions depending on $\theta$),  and
    \item $X=T(\Xt{1},\ldots,\Xt{K})$. 
\end{enumerate}
Then we say that $\mathcal{P}$ is \emph{thinned} by the function $T(\cdot)$.
\end{definition}
When clear from context, sometimes we will say that ``$P_\theta$ is thinned'' or ``X is thinned'', by which we mean that the corresponding family $\mathcal P$ is thinned. 
Intuitively, we can think of thinning as breaking $X$ up into $K$ independent pieces, but in a very particular way that ensures that none of the information about $\theta$ is lost. The fact that no information is lost is evident from the requirement that $X=T(\Xt{1},\ldots,\Xt{K})$. 

Sample splitting \citep{cox1975note} can be viewed as a special case of generalized data thinning. 
\begin{remark}[Sample splitting] \label{remark:samplesplit}
Sample splitting, in which a sample of $n$ independent and identically distributed random variables is partitioned into $K$ subsamples, is a special case of generalized data thinning.  Here, $T(\cdot)$ is the function that takes in the subsamples as arguments, and concatenates and sorts their elements. For more details, see Section~\ref{sec:sample-splitting}.
\end{remark}

Furthermore, Definition~\ref{def:thinning} is closely related to the proposal of \cite{neufeld2023data}.

\begin{remark}[Thinning convolution-closed families of distributions] \label{remark:convclosed}

 \cite{neufeld2023data} show that some well-known families of convolution-closed distributions, such as the binomial, negative binomial, gamma, Poisson, and Gaussian, can be thinned, in the sense of Definition~\ref{def:thinning}, by addition:  
$T\left(\xt{1}, \ldots, \xt{K}\right) = \sum_{k=1}^K \xt{k}$. 
\end{remark}

The two examples above do not resemble each other: The first involves a non-parametric family of distributions and applies quite generally, while the second depends on a  specific property of the family of distributions.  Furthermore, the functions $T(\cdot)$ are quite different from each other.  It is natural to ask: How can we find families $\mathcal P$ that can be thinned?  Is there a unifying principle for the choice of $T(\cdot)$? How can we ensure that there exists a distribution $G_t$ as in Definition~\ref{def:thinning} that does not depend on $\theta$?  The following theorem answers these questions, and indicates that \emph{sufficiency} is the key principle required to ensure that the distribution $G_t$ does not depend on $\theta$.

\begin{theorem}[Main theorem] \label{thm:generalized-thinning}
Suppose $\mathcal P$ is thinned by a function $T(\cdot)$ and, for $X\sim P_\theta$, let $\Qt{1}_\theta\times\cdots\times\Qt{K}_\theta$ denote the distribution of the mutually independent random variables, $(\Xt{1},\ldots,\Xt{K})$, sampled as in Definition~\ref{def:thinning}. Then, the following hold:
\begin{enumerate}[(a)]
    \item $T(\Xt{1},\ldots,\Xt{K})$ is a sufficient statistic for $\theta$ based on $(\Xt{1},\ldots,\Xt{K})$.
    \item The distribution $G_t$  in Definition~\ref{def:thinning} is the conditional distribution
    $$
(\Xt{1},\ldots,\Xt{K})|T(\Xt{1},\ldots,\Xt{K})=t,
    $$
    where $(\Xt{1},\ldots,\Xt{K})\sim \Qt{1}_\theta\times\cdots\times\Qt{K}_\theta$.
    \end{enumerate}
\end{theorem}

Theorem~\ref{thm:generalized-thinning} is proven in Supplement A.1. Further, there is a simple algorithm for finding families of distributions $\mathcal{P}$ and functions $T(\cdot)$ such that $\mathcal P$ can be thinned by $T(\cdot)$.  

\begin{algorithm}[Finding distributions that can be thinned]\label{alg:recipe}
\textcolor{white}{.}
\begin{enumerate}
    \item Choose $K$ families of distributions, $\cQt{k}=\{\Qt{k}_\theta:\theta\in\Omega\}$ for $k=1,\ldots, K$.
    \item Let $(\Xt{1},\ldots,\Xt{K})\sim \Qt{1}_\theta\times\cdots\times\Qt{K}_\theta$, and let  $T(\Xt{1},\ldots,\Xt{K})$ denote a sufficient statistic for $\theta$.
    \item Let $P_\theta$ denote the distribution of $T(\Xt{1},\ldots,\Xt{K})$.  
\end{enumerate}
By construction, the family $\mathcal P=\{P_\theta:\theta\in\Omega\}$ is thinned by $T(\cdot)$.
\end{algorithm}

This recipe gives us a very succinct way to describe the distributions that can be thinned: \emph{We can thin the distributions of sufficient statistics}.
In particular, the recipe takes as input a joint distribution $\Qt{1}_\theta\times\cdots\times\Qt{K}_\theta$, and requires us to choose a sufficient statistic for $\theta$. Then, that statistic's distribution is the $P_\theta$ that can be thinned.

\section{Thinning natural exponential families}
\label{sec:natural-exp-fam}

In Section~\ref{subsec:natural}, we show how to thin a natural exponential family into two or more natural exponential families.  In Section~\ref{subsec:neufeld}, we show how the convolution-closed thinning proposal of \citet{neufeld2023data} can be understood in light of natural exponential family thinning.  Finally, in Section~\ref{subsec:natural-to-general}, we show how natural exponential families can be thinned into more general (i.e., not necessarily natural) exponential families.

\subsection{Thinning natural into natural exponential families} 
\label{subsec:natural}

A natural exponential family \citep{lehmann2005testing} starts with a known probability distribution $H$, and then forms a family of distributions $\mathcal P^H=\{P^H_\theta:\theta\in\Omega\}$ based on $H$:
\begin{equation}
    dP^H_\theta(x)=e^{x^\top\theta-\psi_H(\theta)}dH(x).\label{eq:natural}
\end{equation}
The normalizing constant $e^{-\psi_H(\theta)}$ ensures that $P_\theta$ is a probability distribution, and we take $\Omega$ to be the set of $\theta$ for which this normalization is possible (i.e. for which $\psi_H(\theta)<\infty$).  

The next theorem presents a property of $H$ that is  necessary and sufficient for the resulting natural exponential family $\mathcal P^H$ to be thinned by addition into $K$ natural exponential families.  To streamline the statement of the theorem, we start with a definition.
 
\begin{definition}[$K$-way convolution]
\label{def:conv}
A probability distribution $H$ is the $K$-way convolution of distributions $H_1,\ldots, H_K$ if $\sum_{k=1}^KY_k\sim H$ for $(Y_1,\ldots,Y_K)\sim H_1\times\cdots\times H_K$.
\end{definition}

\begin{theorem}[Thinning natural exponential families by addition]\label{thm:natural-exponential-families}
The natural exponential family $\mathcal P^H$ can be thinned by $T(\xt{1},\ldots,\xt{K})=\sum_{k=1}^K\xt{k}$ into $K$ natural exponential families $\mathcal P^{H_1},\ldots,\mathcal P^{H_K}$ if and only if $H$ is the $K$-way convolution of $H_1,\ldots,H_K$.
\end{theorem}

The $K$ natural exponential families in Theorem \ref{thm:natural-exponential-families} can be different from each other, but they are all indexed by the same $\theta\in\Omega$ that was used in the original family $\mathcal P^H$. The proof of Theorem~\ref{thm:natural-exponential-families} is in Supplement A.2. 

\cite{neufeld2023data} show that it is possible to thin a Gaussian random variable by addition into $K$ independent Gaussians.  We now see that this result follows from Theorem~\ref{thm:natural-exponential-families}. 

\begin{exmp}[Thinning $N_n(\boldsymbol\theta, \mathbf{I}_n)$] 
\label{exmp:gaussian}
Distributions of the form  $N_n(\boldsymbol\theta, \mathbf{I}_n)$ are a natural exponential family indexed by $\boldsymbol\theta\in\mathbb{R}^n$. It can be written in the notation of \eqref{eq:natural}  as 
$\mathcal{P}^{H}$, where $H$ represents the $N_n(\mathbf{0}_n, \mathbf{I}_n)$ distribution.  Furthermore, $H$ is the $K$-way convolution of $H_k = N_n(\mathbf{0}_n, \epsilon_k \mathbf{I}_n)$ for $k = 1, 2, \ldots, K$, where $\epsilon_1,\ldots,\epsilon_K>0$ and $\sum_{k=1}^K \epsilon_k=1$.  Thus, by Theorem~\ref{thm:natural-exponential-families},   we can thin  $\mathcal{P}^{H}$ by addition into $\mathcal P^{H_1}, \ldots, \mathcal P^{H_K}$, where $P_{\boldsymbol\theta}^{H_k}=N_n(\epsilon_k \boldsymbol\theta, \epsilon_k \mathbf{I}_n)$. 
\end{exmp}

In Supplement B, we show that Example \ref{exmp:gaussian} is closely connected 
to a randomization strategy that has been frequently used in the literature.

Not all natural exponential families satisfy the condition of Theorem \ref{thm:natural-exponential-families}. 
We prove in Section \ref{sec:bernoulli} that the distribution $H=\text{Bernoulli}(0.5)$ cannot be written as the sum of two independent, non-constant random variables. 
Since $\mathcal P^{\text{Bernoulli}(0.5)}$ is the $\text{Bernoulli}([1+e^{-\theta}]^{-1})$ natural exponential family, Theorem \ref{thm:natural-exponential-families} implies that Bernoulli random variables cannot be thinned by addition into natural exponential families. In Section~\ref{sec:bernoulli} we will further prove that {\em no function} $T(\cdot)$ can thin the Bernoulli family.

\subsection{Connections to \citet{neufeld2023data}}
\label{subsec:neufeld}

\cite{neufeld2023data} focus on convolution-closed families, i.e., those for which convolving two or more distributions (see Definition~\ref{def:conv}) in the family produces a distribution that is in the family. 
They provide a recipe for decomposing 
a random variable $X$ drawn from a distribution in such a family into independent random variables $X^{(1)}, \ldots, X^{(K)}$ that sum to yield $X$.  
We now show that their results are encompassed by Theorem~\ref{thm:natural-exponential-families}.

\emph{Exponential dispersion families} \citep{jorgensen1992exponential, jorgensen1998stationary} are a subclass of convolution-closed families. Given a distribution $H$ with $\psi_H(\theta)<\infty$ for $\theta\in\Omega$ (as in \eqref{eq:natural}), we identify the set of distributions $H_\lambda$ for which $\psi_{H_\lambda}(\cdot)=\lambda\psi_H(\cdot)$ (i.e., distributions whose cumulant generating function is a multiple of $H$'s cumulant generating function).  We define $\Lambda$ to be the set of $\lambda$ for which such a distribution $H_\lambda$ exists.  Then, an (additive) \emph{exponential dispersion family} is $\mathcal{P} = \bigcup_{\lambda \in \Lambda}\mathcal P^{H_\lambda}$, where $\mathcal P^{H_\lambda}$ is the natural exponential family generated by $H_\lambda$ (see \eqref{eq:natural}).  The distributions in $\mathcal P$ are indexed over $(\theta,\lambda)\in\Omega\times\Lambda$ and take the form $dP^{H_\lambda}_{\theta} (x) = e^{x^\top \theta - \lambda \psi_H(\theta)} dH_\lambda(x)$. 

In words, an exponential dispersion family results from  combining a collection of related natural exponential families.  For example, starting with $H=\text{Bernoulli(1/2)}$, we can take $\Lambda=\mathbb Z^+$ since for any positive integer $\lambda$, $H_\lambda=\text{Binomial}(\lambda,1/2)$ satisfies the necessary cumulant generating function relationship.  Then, $\mathcal P^{\text{Binomial}(\lambda,1/2)}$ corresponds to the binomial natural exponential family that results from fixing $\lambda$. 
Finally, allowing $\lambda$ to vary gives the full binomial exponential dispersion family, which is the set of all binomial distributions (varying both of the parameters of the binomial distribution).

By construction, for any $\lambda_1,\ldots,\lambda_K\in\Lambda$, convolving $P^{H_{\lambda_1}}_\theta,\ldots,P^{H_{\lambda_K}}_\theta$ gives the distribution $P^{H_\lambda}_\theta$, where $\lambda=\sum_{k=1}^K\lambda_k$.  The next corollary is an immediate application of Theorem~\ref{thm:natural-exponential-families} in the context of exponential dispersion families.
Notably, 
 the distributions $\Qt{k}_\theta$ themselves still belong to the exponential dispersion family $\mathcal P$ to which the distribution of $X$ belongs. 
\begin{corollary}[Thinning while remaining inside an exponential dispersion family] \label{cor:neufeld}
Consider an exponential dispersion family $\mathcal{P} = \bigcup_{\lambda \in \Lambda}\mathcal P^{H_\lambda}$ and suppose $\lambda_1,\ldots,\lambda_K\in\Lambda$.  Then for $\lambda=\sum_{k=1}^K\lambda_k$, we can thin the natural exponential family $\mathcal P^{H_\lambda}$ by $T(\xt{1},\ldots,\xt{K})=\sum_{k=1}^K\xt{k}$ into the natural exponential families $\mathcal P^{H_{\lambda_1}},\ldots,\mathcal P^{H_{\lambda_K}}$.
\end{corollary}
This result corresponds exactly to the data thinning proposal of  \cite{neufeld2023data}.  We see from Corollary~\ref{cor:neufeld} that that proposal thins a natural exponential family, $\mathcal P^{H_\lambda}$, into a \emph{different} set of natural exponential families, 
$\mathcal P^{H_{\lambda_1}},\ldots,\mathcal P^{H_{\lambda_K}}$. However, 
 from the perspective of exponential dispersion families, it thins an exponential dispersion family into the same exponential dispersion family.  Continuing the binomial example from above, the corollary tells us that we can thin the binomial family with $\lambda$ as the number of trials into two or more binomial families with smaller numbers of trials, provided that $\lambda>1$.

\cite{neufeld2023data} focus on convolution-closed families, not exponential dispersion families. However, all convolution-closed families that have moment-generating functions can be written as exponential dispersion families \citep{jorgensen1998stationary}. The Cauchy distribution is convolution-closed, but does not have a moment generating function and thus is not an exponential dispersion family. As we will see in Example~\ref{ex:cauchy}, the $\text{Cauchy}(\theta_1, \theta_2)$ distribution cannot be thinned by addition: Decomposing it using the recipe of \citet{neufeld2023data} requires knowledge of both unknown parameters. Thus, not all convolution-closed distributions can be thinned by addition in the sense of Definition~\ref{def:thinning}. However, \cite{neufeld2023data} claim that all convolution-closed distributions \emph{can} be thinned. This apparent discrepancy  is due to a slight difference in the definition of thinning between our paper and theirs: Our Definition~\ref{def:thinning} requires that $G_t$ not depend on $\theta$; however, \cite{neufeld2023data} have no such requirement. In practice, data thinning is useful only if $G_t$ does not depend on $\theta$, and so there is no meaningful difference between the two definitions.

\subsection{Thinning natural into general exponential families} 
\label{subsec:natural-to-general}

In this section, we apply Algorithm~\ref{alg:recipe} in the case that $\cQt{k}$ are (possibly non-natural) exponential families, for which the sufficient statistic need not be the identity. 
In particular, for $k=1,\ldots, K$, we let $\cQt{k}=\{\Qt{k}_\theta:\theta\in\Omega\}$ denote an exponential family based on a known distribution $H_k$ and sufficient statistic $T^{(k)}(\cdot)$:
\begin{equation}\label{eq:exp-fam}
d\Qt{k}_\theta(x)=\exp\{[T^{(k)}(x)]^\top\eta(\theta)-\psi_k(\theta)\}dH_k(x).
\end{equation}
As in Section \ref{subsec:natural}, $e^{-\psi_k(\theta)}$ is the normalizing constant needed to ensure that $\int d\Qt{k}_\theta(x)=1$ and $\Omega$ is the set of $\theta$ for which $\psi_k(\theta)<\infty$.  The function $\eta(\cdot)$ maps $\theta$ to the natural parameter.  We note that  $\sum_{k=1}^K T^{(k)}(\Xt{k})$ is a sufficient statistic for $\theta$ based on $(\Xt{1},\dots,\Xt{K})\sim \Qt{1}_\theta\times\cdots\times\Qt{K}_\theta$.   Then,  Algorithm~\ref{alg:recipe} tells us that we can thin the distribution of this sufficient statistic. This  leads to the next result.

\begin{proposition}[Thinning natural exponential families with more general functions $T(\cdot)$] \label{prop:exp-family}
Let $\Xt{1},\dots,\Xt{K}$ be independent random variables with $\Xt{k}\sim\Qt{k}_\theta$ for $k=1,\ldots,K$ from any (i.e., possibly non-natural) exponential families $\cQt{k}$ as in \eqref{eq:exp-fam}.
Let $P_\theta$ denote the distribution of $\sum_{k=1}^K T^{(k)}(\Xt{k})$. Then, $\mathcal P=\{P_\theta:\theta\in\Omega\}$ is a natural exponential family, and we can thin it  into $\Xt{1},\ldots,\Xt{K}$ using the function $T(\xt{1},\dots,\xt{K})=\sum_{k=1}^K T^{(k)}(\xt{k})$.
\end{proposition}
The fact that $\mathcal P$ in this result is a natural exponential family follows from recalling that the sufficient statistic of an exponential family follows a natural exponential family \citep[Lemma~2.7.2(i)]{lehmann2005testing}. 
Many named exponential families are not natural exponential families, involving non-identity functions $T^{(k)}(\cdot)$, such as the logarithm or polynomials.  Therefore, to thin into those families, Proposition~\ref{prop:exp-family} will be useful.

Proposition~\ref{prop:exp-family} implies that many natural exponential families \textit{can} be thinned by a function of the form $T(\xt{1},\dots,\xt{K})=\sum_{k=1}^K T^{(k)}(\xt{k})$. Theorem~\ref{thm:backwards-natexpfam} shows that if a full-rank natural exponential family can be thinned, then the thinning function \textit{must} take this form. 

\begin{theorem}[Thinning functions for natural exponential families] \label{thm:backwards-natexpfam}
Suppose $X\sim P_\theta$, where $\mathcal{P}=\{P_\theta: \theta\in\Omega\}$ is a full-rank natural exponential family with density/mass function $p_\theta(x)=\exp(\theta^\top x - \psi(\theta))h(x)$. If $\mathcal{P}$ can be thinned by $T(\cdot)$ into $\Xt{1},\dots,\Xt{K}$, then:
\begin{enumerate}
    \item The function $T(\xt{1},\dots,\xt{K})$ is of the form $\sum_{k=1}^K T^{(k)}(\xt{k})$.
    \item $\Xt{k}\overset{\text{ind}}{\sim} \Qt{k}_\theta$ where $\Qt{k}_\theta$ is an exponential family with sufficient statistic $T^{(k)}(\Xt{k})$.
\end{enumerate}
\end{theorem}

The proof of Theorem~\ref{thm:backwards-natexpfam} is provided in Supplement A.3. 

To illustrate the flexibility provided by Proposition~\ref{prop:exp-family} and Theorem~\ref{thm:backwards-natexpfam}, we demonstrate that a natural exponential family $\mathcal P$ can be thinned by different functions $T(\cdot)$, leading to families of distributions $\cQt{1},\ldots,\cQt{K}$ different from $\mathcal P$. Specifically, we consider three possible $K$-fold thinning strategies for a gamma distribution when the shape, $\alpha$, is known but the rate\footnote{Although $\theta$ is often used in the gamma distribution to denote the scale parameter, here we use it to denote the rate parameter.}, $\theta$, is unknown.  

\begin{exmp}[Thinning $\text{Gamma}(\alpha,\theta)$ with $\alpha$ known, approach 1] \label{ex:dtgamma} 
Following Algorithm~\ref{alg:recipe}, we start with $\Xt{k} \overset{iid}{\sim} \text{Gamma}\left(\frac{\alpha}{K},\theta\right)$ for $k=1,\ldots,K$, and note that $T(\Xt{1},\dots,\Xt{K})=\sum_{k=1}^K\Xt{k}$ is sufficient for $\theta$. Thus, we can thin the distribution of $\sum_{k=1}^K\Xt{k}$.  A well-known property of the gamma distribution tells us that this is a $\text{Gamma}(\alpha,\theta)$ distribution.  Sampling from $G_t$ as in Theorem~\ref{thm:generalized-thinning} 
corresponds exactly to the multi-fold gamma data thinning recipe of \cite{neufeld2023data} where $\epsilon_k=\frac{1}{K}$.
\end{exmp}

Alternatively, when $\alpha$ can be expressed as half of a natural number, we can apply Proposition~\ref{prop:exp-family} to decompose the gamma family into centred normal data. 

\begin{exmp}[Thinning $\text{Gamma}(\alpha,\theta)$ with $\alpha=K/2$ known, approach 2] \label{ex:scaled-normal}
Starting with $\Xt{k} \overset{iid}{\sim} N(0,\frac{1}{2\theta})$, notice that $T^{(k)}(\xt{k})=(\xt{k})^2$. We thus apply Proposition~\ref{prop:exp-family} using $T(\xt{1},\dots,\xt{K})=\sum_{k=1}^K(\xt{k})^2$ to thin the sufficient statistic, $\sum_{k=1}^K(\Xt{k})^2\sim\frac{1}{2\theta}\chi_K^2=\text{Gamma}\left(\frac{K}{2},\theta\right)$, into $(\Xt{1},\dots,\Xt{K})$. The function $G_t$ from Theorem~\ref{thm:generalized-thinning} is the conditional distribution $(\Xt{1},\dots,\Xt{K})|\sum_{k=1}^K(\Xt{k})^2=t$. By rotational symmetry of the $N_K(0,(2\theta)^{-1}\mathbf{I}_K)$ distribution (the joint distribution of $(\Xt{1},\dots,\Xt{K})$), $G_t$ is the uniform distribution on the $(K-1)$-sphere of radius $t^{1/2}$. To sample from this conditional distribution, we generate $\mathbf{Z}\sim N_K(0,\mathbf{I}_K)$ and then take  $(\Xt{1},\dots,\Xt{K})$ to be $t^{1/2}\frac{\mathbf{Z}}{\|\mathbf{Z}\|_2}$.
\end{exmp}
If $\alpha$ is a natural number, then applying a similar logic enables us to thin the gamma family with unknown rate into the Weibull family with unknown scale; see Example C.1 in Supplement C.1.1. From a theoretical perspective, when $\alpha$ is a natural number, there is no reason to prefer one of the  three gamma thinning strategies over another. However, there may be practical considerations: For instance,  the strategy  in Example \ref{ex:scaled-normal} may be preferred due to the convenience of working with Gaussian data. In general, if multiple thinning strategies are available, then the choice  can be driven by modeling convenience.

\section{Indirect thinning of general exponential families} \label{sec:general-exp}

Sometimes rather than thinning $X$, we may choose  to thin a function $S(X)$.  When $S(X)$ is sufficient for $\theta$ based on $X$, the next proposition tells us that  thinning $S(X)$ rather than $X$ does not result in a loss of  information about $\theta$.  We emphasize that we are using the concept of sufficiency in two ways  here: (i) $S(X)$ is sufficient for $\theta$ based on $X\sim P_\theta$, and (ii)  $T(\Xt{1},\ldots,\Xt{K})$  is sufficient for $\theta$ based on $(\Xt{1},\ldots,\Xt{K})\sim \Qt{1}_\theta\times\cdots\times\Qt{K}_\theta$.  

\begin{proposition}[Thinning a sufficient statistic preserves information]\label{prop:thin-S}
Suppose $X\sim P_\theta\in\mathcal P$ has a sufficient statistic $S(X)$ for $\theta$, and we thin $S(X)$ by  $T(\cdot)$. That is, conditional on $S(X)$ (and without knowledge of $\theta$) we sample $\Xt{1},\ldots,\Xt{K}$ that are mutually independent and satisfy
$S(X)=T(\Xt{1},\ldots,\Xt{K})$.
Under regularity conditions 
needed for Fisher information to exist, we have that
$I_X(\theta)=\sum_{k=1}^K I_{\Xt{k}}(\theta)$.
\end{proposition}

This proposition shows that thinning $S(X)$, rather than $X$, does not result in a loss of information about $\theta$.  
Its proof (provided in Supplement A.4) follows easily from multiple applications of the fact that sufficient statistics preserve information.
Definition~\ref{def:indirect} formalizes the strategy suggested by Proposition~\ref{prop:thin-S}.

\begin{definition}[Indirect thinning] 
\label{def:indirect}
Consider $X\sim P_\theta\in \mathcal P$. Suppose we thin a sufficient statistic $S(X)$ for $\theta$ by a function $T(\cdot)$. 
We say that the family $\mathcal P$ is {\em indirectly thinned through $S(\cdot)$ by $T(\cdot)$}. 
\end{definition}
In light of Proposition \ref{prop:thin-S}, indirect thinning does not result in a loss of information.

When $S(\cdot)$ is invertible, then
$X=S^{-1}(T(\Xt{1},\ldots,\Xt{K}))$,
which implies that we can thin $X$ directly by $S^{-1}(T(\cdot))$.
It turns out that, regardless of whether we thin $X$ by $S^{-1}(T(\cdot))$ or indirectly thin $X$ through $S(\cdot)$ by $T(\cdot)$, there is little difference between the resulting form of $G_t$ in Theorem~\ref{thm:generalized-thinning}. In the former case, $G_t$ is the conditional distribution of $(\Xt{1},\dots,\Xt{K})$ given $S^{-1}( T(\Xt{1},\dots,\Xt{K}))=  t$. In the  latter case, it is  the conditional distribution of $(\Xt{1},\dots,\Xt{K})$ given $ T(\Xt{1},\dots,\Xt{K})=  t$. Since $S(\cdot)$ is invertible, these two conditional distributions are identical following a reparameterization. 

We now return to the setting of Proposition \ref{prop:thin-S}, where $S(\cdot)$ may or may not be invertible.

\begin{remark}[Indirect thinning of general exponential families]
\label{remark:genexp}
Let $\mathcal{P} = \{P_\theta: \theta \in \Omega\}$ be a full-rank general exponential family. That is, $dP_\theta(x) = \exp\{[S(x)]^\top\eta(\theta) - \psi(\theta)\}dH(x)$,
where $e^{-\psi(\theta)}$ is the normalising constant. Since $S(X)$ is sufficient for $\theta$, we can  indirectly thin $X$ through $S(\cdot)$ without a loss of Fisher information (Proposition~\ref{prop:thin-S}). 
Furthermore,  $S(X)$ belongs to a full-rank natural exponential family \citep[Lemma~2.7.2(i)]{lehmann2005testing}. We can thus indirectly thin $X$ through $S(\cdot)$ as follows: 
\begin{enumerate}
\item Provided that the necessary and sufficient condition of Theorem \ref{thm:natural-exponential-families} holds for $S(X)$, we can indirectly thin $X$ through $S(\cdot)$ by addition into $\Xt{1},\ldots,\Xt{K}$ that follow natural exponential families, i.e. \eqref{eq:exp-fam} where $T^{(k)}(\cdot)$ is the identity.
\item  We now consider $\Xt{1},\ldots,\Xt{K}$ that belong to a general exponential family, where
 $T^{(k)}(\cdot)$ in \eqref{eq:exp-fam} is not necessarily the identity. Suppose further that
$S(X)\overset{D}{=}\sum_{k=1}^KT^{(k)}(\Xt{k})$. 
 Then, by Proposition~\ref{prop:exp-family}, we can indirectly thin $X$ through $S(\cdot)$ into 
$\Xt{1},\ldots,\Xt{K}$,
by  $T(\xt{1},\ldots,\xt{K})=\sum_{k=1}^KT^{(k)}(\xt{k})$.
  
\end{enumerate}
We see that 1) is a special case of 2). 
\end{remark}

We now demonstrate indirect thinning with some examples.   
First, we consider a $\text{Beta}(\theta,\beta)$ random variable, with $\beta$ a known parameter. This is not a natural exponential family, and so the results in Section \ref{sec:natural-exp-fam} are not directly applicable.  The beta family also differs from the other examples that we have seen in the following ways:  (i)  It is not convolution-closed; (ii) it has finite support; and (iii) the sufficient statistic for an independent and identically distributed sample has an unnamed distribution. 

\begin{exmp}[Thinning $\text{Beta}(\theta,\beta)$ with $\beta$ known] \label{ex:beta}
We start with $\Xt{k} \overset{ind}{\sim} \text{Beta}\left(\frac{1}{K}\theta + \frac{k-1}{K}, \frac{1}{K}\beta\right)$, for  $k=1,\dots,K$; this is a 
general exponential family \eqref{eq:exp-fam} with $T^{(k)}(\xt{k}) = \frac1{K}\log(\xt{k})$. Since $ \sum_{k=1}^K T^{(k)}(\Xt{k})$ is sufficient for $\theta$ based on $\Xt{1},\ldots,\Xt{K}$, we can apply Proposition~\ref{prop:exp-family} to thin the distribution of $\sum_{k=1}^K T^{(k)}(\Xt{k})$ by the function
\begin{equation}\label{eq:beta}
T(\xt{1},\dots,\xt{K})=\sum_{k=1}^K T^{(k)}(\xt{k})=\frac1{K}\sum_{k=1}^K\log(\xt{k})=\log\left[\left(\prod_{k=1}^K\xt{k}\right)^{1/K}\right]. 
\end{equation}

Furthermore, we show in Supplement C.1.2 that $\exp\left( T(\Xt{1},\dots,\Xt{K})\right) =\left(\prod_{k=1}^K\Xt{k}\right)^{1/K}$, the geometric mean of $\Xt{1},\ldots,\Xt{K}$,   
follows a $\text{Beta}(\theta, \beta)$ distribution. Therefore, we can indirectly thin a $\text{Beta}(\theta, \beta)$  random variable through  $S(x)=\log(x)$ by $T(\cdot)$ defined in \eqref{eq:beta}. This results in $\Xt{k} \overset{ind}{\sim} \text{Beta}\left(\frac{1}{K}\theta + \frac{k-1}{K}, \frac{1}{K}\beta\right)$, for $k=1,\ldots,K$. 

Furthermore, since $S(x)=\log(x)$ is invertible,  we can \emph{directly} thin $X\sim\text{Beta}(\theta, \beta)$  by 
\begin{equation}
    \label{eq:betaprime}
    T'(\xt{1},\dots,\xt{K})= S^{-1} (T(  \xt{1},\dots,\xt{K})) = \left(\prod_{k=1}^K\xt{k}\right)^{1/K}.
\end{equation}

To apply either of these thinning strategies, we need to sample from  $G_t$ defined in Theorem~\ref{thm:generalized-thinning}. This can be done using numerical methods, as detailed in Supplement C.1.2.  
\end{exmp}  

By symmetry of the beta distribution, we can also apply the thinning operations detailed in Example~\ref{ex:beta} to thin a $\text{Beta}(\alpha, \theta)$ random variable with $\alpha$ known. In Example C.2 in Supplement C.1.3, we propose an alternative strategy to thin a beta random variable, using a different parametrization. As this example extends naturally to higher dimensions, we derive and prove it for the more general Dirichlet case.

Next, we consider thinning the gamma distribution with unknown shape parameter.

\begin{exmp}[Thinning $\text{Gamma}(\theta,\beta)$ with $\beta$ known] \label{ex:gamma1} 
We start with $\Xt{k} \overset{ind}{\sim} \text{Gamma}\left(\frac{1}{K}\theta + \frac{k-1}{K}, \frac{1}{K}\beta\right)$,  for  $k=1,\dots,K$; this is a general exponential family (\ref{eq:exp-fam}) with $T^{(k)}(\Xt{k})=\frac{1}{K}\log(\xt{k})$. Note that
$T(\Xt{1},\dots,\Xt{K})=\sum_{k=1}^K T^{(k)}(\Xt{k})$ is sufficient for $\theta$ based on $\Xt{1},\ldots,\Xt{K}$. As $T^{(k)}(\cdot)$ is shared with Example \ref{ex:beta}, we can apply Proposition \ref{prop:exp-family} to thin the distribution of $\sum_{k=1}^K T^{(k)}(\Xt{k})$ by the function defined in (\ref{eq:beta}).

In Supplement C.1.4 we show that $\exp\left(T(\Xt{1},\dots,\Xt{K}\right)=\left(\prod_{k=1}^K\Xt{k}\right)^{1/K}$
follows a $\text{Gamma}(\theta, \beta)$ distribution. Thus, we can indirectly thin a $\text{Gamma}(\theta,\beta)$ random variable through $S(x)=\log(x)$ by $T(\cdot)$ defined in (\ref{eq:beta}). This produces independent random variables $\Xt{k}\sim\text{Gamma}(\frac{1}{K}\theta+\frac{k-1}{K},\frac{1}{K}\beta)$ for $k=1,\dots,K$. Once again, noting that $S(\cdot)$ is invertible, we can instead directly thin $X\sim\text{Gamma}(\theta,\beta)$ by the function defined in \eqref{eq:betaprime}.

To apply either of these thinning strategies to a $\text{Gamma}(\theta,\beta)$ random variable, we must sample from $G_t$ as defined in Theorem \ref{thm:generalized-thinning}. 
See Supplement C.1.4. 
\end{exmp}

 Example~\ref{ex:gamma1} is different from the gamma thinning example from \cite{neufeld2023data}: That involves thinning a $\text{Gamma}(\alpha,\theta)$ random variable with $\alpha$ known, whereas here we thin a $\text{Gamma}(\theta,\beta)$ random variable with $\beta$ known. 

Examples \ref{ex:beta} and \ref{ex:gamma1} enable us to thin a random variable into an arbitrary number of independent random variables. However, unlike in the examples in Section~\ref{sec:natural-exp-fam}, the  resulting folds are not identically distributed.

In Examples~\ref{ex:beta} and \ref{ex:gamma1}, the function  $S(\cdot)$ through which we indirectly thin $X$ is invertible. Supplement D considers indirect thinning of a sample of $n$ independent and identically distributed normal random variables with both mean and variance unknown. This provides an example of a case
in which $S(\cdot)$ is neither invertible, nor scalar-valued.

We close with a list of a few short examples to illustrate the flexibility of indirect thinning. 

\begin{exmp}[Additional examples of indirect thinning]
\label{ex:other}
\textcolor{white}{.}
\begin{enumerate}
    \item Suppose we observe $X\sim N(\mu,\theta)$ where $\mu$ is known; here $\mu$ denotes the mean and $\theta$ the variance. Then $S(X)=(X-\mu)^2\sim\theta\chi_1^2=\text{Gamma}\left(\frac{1}{2},\frac{1}{2\theta}\right)$. Thus, by applying the Gamma thinning strategy of \cite{neufeld2023data} discussed in Example \ref{ex:dtgamma} to $S(X)$, we can indirectly thin a normal distribution with unknown variance through $S(\cdot)$. \label{ex:normal-unknown-variance}
    \item Suppose we observe $X\sim\text{Weibull}(\theta, \gamma)$ where $\gamma$ is known. Then, $S(X)=X^\gamma\sim\text{Exp}\left(\theta^{-\gamma}\right)$. Thus, by applying the Gamma thinning strategy of Example \ref{ex:dtgamma} or \ref{ex:scaled-normal} to $S(X)$, we can indirectly thin a Weibull distribution with unknown rate through $S(\cdot)$.
    \item Suppose we observe $X\sim\text{Pareto}(\gamma,\theta)$ where $\gamma$ is known. Then $S(X)=\log \left( X/\gamma \right) \sim\text{Exp}(\theta)$. Thus, by applying the Gamma thinning strategy of Example \ref{ex:dtgamma} or \ref{ex:scaled-normal} to $S(X)$, we can indirectly thin a Pareto distribution with unknown shape through $S(\cdot)$. 
\end{enumerate}
\end{exmp}

\section{Thinning outside of exponential families}
\label{sec:outside-exp-fam}

In this section, we focus on thinning outside of exponential families.
Outside of the exponential family, only certain distributions with domains that vary with the parameter of interest have sufficient statistics that are bounded as the sample size increases \citep{darmois, koopman, pitman_1936}. 
Thus, we first consider a setting where $\theta$ alters the support of the distribution (Section~\ref{sec:varysupport}), and then one where the sufficient statistic's dimension grows as the sample size increases (Section~\ref{sec:sample-splitting}). 

\subsection{Thinning distributions with varying support} \label{sec:varysupport}

We consider examples in which the parameter of interest, $\theta$, changes the support of a distribution. 
In Example~\ref{ex:scaled-uniform}, $\theta$ scales the support.  

\begin{exmp}[Thinning $\text{Unif}(0,\theta)$] \label{ex:scaled-uniform}
We start with  $\Xt{k} \overset{iid}{\sim} \theta \cdot\text{Beta}\left(\frac{1}{K},1\right)$ for $k=1,\ldots,K$, and note that $T(\Xt{1},\dots,\Xt{K}) = \max(\Xt{1},\dots,\Xt{K})$ is sufficient for $\theta$. 
 Furthermore, $\max(\Xt{1},\dots,\Xt{K}) \sim \text{Unif}(0, \theta)$. Thus, we define $G_t$ to be the conditional distribution of 
 $(\Xt{1},\dots,\Xt{K})$ given $\max(\Xt{1},\dots,\Xt{K}) = t$. Then, 
 by Theorem~\ref{thm:generalized-thinning}, 
 we can thin  $X \sim \text{Unif}(0, \theta)$ by 
 sampling from  $G_X$. 
To do this, we first   
draw $\mathbf{C}\sim\text{Categorical}_K\left(1/K, \ldots, 1/K\right)$. Then, $\Xt{k} = C_kX+(1-C_k)Z_k$ where $Z_k \overset{iid}{\sim} X\cdot\text{Beta}\left(\frac{1}{K},1\right)$.
\end{exmp}
This is a special case of Example C.3 in Supplement C.2.1, in which we thin the scale family $\theta\cdot\text{Beta}(\alpha,1)$ where $\alpha$ is known. Setting $\alpha=1$ yields Example \ref{ex:scaled-uniform}.

Similar thinning results can be identified for distributions in which $\theta$ shifts the support. In Supplement C.2.2, we show that $X\sim \mathrm{SExp}(\theta,\lambda)$, the location family generated by shifting an exponential random variable by $\theta$, can be thinned by the minimum function.

\subsection{Sample splitting as a special case of generalized data thinning}\label{sec:sample-splitting}

We now consider sample splitting, a well-known approach for splitting a sample of  observations into two or more sets \citep{cox1975note}. 
We show that sample splitting can be viewed as an instance of generalized data thinning. In this setting, $\mathbf{X}=(X_1,\ldots,X_n)$ is a sample of independent and identically distributed random variables, $X_i\in\mathcal X$, each having distribution $F\in\mathcal F$, where $\mathcal F$ is some (potentially non-parametric) family of distributions and $\mathcal X$ is the set of values that the random variable $X_i$ can take (most commonly $\mathcal X=\mathbb R^p$).  That is, $\mathbf{X}\sim P_F\in \mathcal P$, where
$\mathcal P=\{F^n:F\in\mathcal F\}$,
and $F^n=F\times\cdots\times F$ denotes the joint distribution of $n$ independent random variables drawn from $F$.

\begin{exmp}[Sample splitting is a special case of generalized data thinning] \label{ex:samplesplit}
We begin with $\mathbf{X}^{(k)} := (\Xt{k}_{1},\ldots,\Xt{k}_{n_k}) \overset{iid}{\sim} F^{n_k}$, for $k=1,\ldots,K$. Here, 
$n_1,\ldots,n_K >0$, 
 and $\sum_{k=1}^K n_k=n$. 
 That is, 
for $k=1,\ldots,K$, $\mathbf{X}^{(k)}\in\mathcal X^{n_k}$ denotes a set of $n_k$ independent and identically distributed draws from $F$.

Our goal is to thin $S(\mathbf{X})$, where $S:\mathcal X^n\to\mathcal X^n$ sorts the entries of its input based on their values.
We define $T:\mathcal X^{n_1}\times\cdots\times \mathcal X^{n_K}\to \mathcal X^n$ as
$T(\mathbf{x}^{(1)},\ldots,\mathbf{x}^{(K)})=S((\mathbf{x}^{(1)},\ldots,\mathbf{x}^{(K)}))$,
the function that  concatenates its arguments and then applies $S(\cdot)$. 
Then $T(\mathbf{X}^{(1)},\ldots,\mathbf{X}^{(K)})$ is a sufficient statistic for $F$, and 
furthermore, $T(\mathbf{X}^{(1)},\ldots,\mathbf{X}^{(K)}) \overset{D}{=} S(\mathbf{X})$. 

We define $G_{\mathbf{t}}$ to be the conditional distribution of $(\mathbf{X}^{(1)},\ldots,\mathbf{X}^{(K)})$ given $  T(\mathbf{X}^{(1)},\ldots,\mathbf{X}^{(K)})=\mathbf{t}$. Suppose we observe $\mathbf{X}\sim F^n$. 
Then, by Theorem~\ref{thm:generalized-thinning}, we can indirectly thin $\mathbf{X}$ through $S(\cdot)$ by $T(\cdot)$ by sampling from $G_{S(\mathbf{X})}$. This conditional distribution is uniform over all   
$\frac{n!}{n_1!\cdots n_K!}$ assignments of $n$ items to $K$ groups of sizes $n_1,\ldots,n_K$. Thus, to sample from $G_{S(\mathbf{X})}$, we randomly partition the sample of size $n$ into $K$ groups of sizes $n_1,\ldots,n_K$.
This is precisely the same as sample splitting. 
\end{exmp}

We have shown that when one has $n$ independent and identically distributed samples from a distribution $F$, then sample splitting is an instance of generalized data thinning. When this assumption holds, it follows from Proposition \ref{prop:thin-S} that sample splitting preserves all information about $F$.  In practice, however, sample splitting is often applied in situations where we have $n$ random variables that are not independent or not identically distributed. In such a situation, using a valid generalized data thinning strategy will be advantageous. For example, consider the setting of multivariate Gaussian data with known dense covariance. Since the data are not independent, sample splitting will produce dependent folds whereas multivariate Gaussian data thinning generates independent folds. Next, consider the case of linear regression with a fixed design matrix: The data are independent but not identically distributed. In this setting, \cite{neufeld2023data} and \cite{rasines2021splitting} show that Gaussian data thinning is preferable to sample splitting from the standpoint of Fisher information (see Section 4 of \cite{neufeld2023data} for technical details).
\section{Counterexamples}
\label{sec:counterexamples}

We now present two examples in which thinning strategies do not work.
The first  involves a natural exponential family that is based on a distribution that {\em cannot} be written as the convolution of two distributions. In this case, Theorem~\ref{thm:natural-exponential-families} implies that we cannot thin it by addition. In fact, we will prove a stronger statement: Namely, that  there does not exist \emph{any} function $T(\cdot)$ that can thin it.  The second example involves a convolution-closed family outside of the natural exponential family in which addition is not sufficient. In this case, taking $T(\cdot)$ to be addition does not enable thinning, as Theorem \ref{thm:generalized-thinning} does not apply.

\subsection{The Bernoulli family cannot be thinned}
\label{sec:bernoulli}

Let $P_{\theta}$ denote the $\mathrm{Bernoulli}(\theta)$ distribution, where $\theta$ is the probability of success. Recall that this distribution can be written as a natural exponential family (with natural parameter $\log \left( \frac{\theta}{1 - \theta} \right)$). By Theorem \ref{thm:backwards-natexpfam}, if $P_\theta$ can be thinned, then the thinning function $T(\cdot)$ must be  additive. However, as the next theorem shows, the Bernoulli distribution cannot be written as a convolution of independent, non-constant random variables.

\begin{theorem}[The Bernoulli is not a convolution]\label{thm:bernoulli}
If $Z^{(1)}$ and $Z^{(2)}$ are independent, non-constant random variables, then $Z^{(1)}+Z^{(2)}$
cannot be a Bernoulli random variable.
\end{theorem}

Theorem~\ref{thm:bernoulli} is proven in Supplement A.5. 

As the Bernoulli distribution cannot be written as a convolution of non-constant random variables, it cannot achieve the two conclusions of Theorem~\ref{thm:backwards-natexpfam} simultaneously. Thus, a contrapositive argument applied to Theorem~\ref{thm:backwards-natexpfam} leads to the next result.

\begin{corollary}
\label{cor:bern2}
    The Bernoulli family cannot be thinned by any function $T(\cdot)$.
\end{corollary}

This corollary of Theorems \ref{thm:backwards-natexpfam} and \ref{thm:bernoulli} is proven in Supplement A.6.
A similar argument reveals that the categorical distribution also cannot be thinned.

The above corollary pertains to a \emph{single} Bernoulli random variable. By contrast, a \emph{vector} of independent and identically distributed Bernoulli random variables can be thinned by sample splitting or by indirect binomial thinning on the sum of the  entries.

\subsection{The Cauchy family cannot be thinned by addition}
\label{subsec:cauchy}

Suppose now that our interest lies in a random variable $X=T(\Xt{1},\Xt{2})$, where $T(\Xt{1}, \Xt{2})$ is \emph{not} sufficient for the parameter $\theta$ based on $(\Xt{1},\Xt{2})$. This means that the conditional distribution of $(\Xt{1},\Xt{2})$ given $T(\Xt{1},\Xt{2})$ depends on $\theta$, and thus that we cannot thin $X$ by $T(\cdot)$.  We see this in the following example.

\begin{exmp}[The trouble with thinning $\text{Cauchy}(\theta_1, \theta_2)$ by addition] \label{ex:cauchy}
Recall that the Cauchy family, $\text{Cauchy}(\theta_1, \theta_2)$, indexed by $\boldsymbol\theta=(\theta_1,\theta_2)$,  is convolution-closed. In particular, if $\Xt{1},\Xt{2}\overset{iid}{\sim}\text{Cauchy}\left(\frac{1}{2}\theta_1, \frac{1}{2}\theta_2\right)$, then $\Xt{1} + \Xt{2}\sim\text{Cauchy}(\theta_1, \theta_2)$.  It is tempting therefore to try thinning this family by $T(\xt{1},\xt{2})=\xt{1}+\xt{2}$.  However, the sum $\Xt{1} + \Xt{2}$ is not sufficient for either 
$\theta_1$ or $\theta_2$, which means that Theorem~\ref{thm:generalized-thinning} does not apply. In particular, $G_t$, the conditional distribution of $(\Xt{1},\Xt{2})$ given $\Xt{1} + \Xt{2}=t$, \emph{is} a function of $\boldsymbol\theta$. Therefore, we cannot thin the Cauchy family with any unknown parameters by addition. 
\end{exmp}

We can take this result a step further: Given a collection of Cauchy random variables, there is no sufficient statistic for $\boldsymbol\theta$ that reduces the data beyond the order statistics
\citep[p.~275]{CaseBerg}. 
Thus, following Algorithm \ref{alg:recipe} with $\cQt{k}$ being $\text{Cauchy}(\theta_1,\theta_2)$, the only generalized data thinning approach that generates independent Cauchy random variables is sample splitting a vector of independent Cauchy random variables. 

\section{Changepoint detection in wind speed data}
\label{sec:changepoint}

To demonstrate the utility of generalized data thinning, we consider detecting changepoints in the variance of wind speed data. We consider a wind speed dataset \citep{haslett1989space} collected in the Irish town of Claremorris, available in the R package \texttt{gstat} \citep{pebesma2004multivariable}. 
\citet{killick2014changepoint} took first differences to remove the periodic mean, and then modeled the
resulting $X_i$ for $i = 1, \dots, n$ as independent normal observations with $X_i \sim N(0, \theta_i)$. They then estimated changepoints in the variance $\theta_1,\ldots,\theta_n$. Here, we take their analysis a step further by testing for a difference in variance on either side of each estimated changepoint.

First, we consider a naive approach.

\begin{algorithm}[Naive approach for changepoint detection]\label{alg:naive}
\textcolor{white}{.}
\begin{enumerate}
    \item Compute $Z_i := X_i^2$. Note that $Z_i  \sim \text{Gamma}\left(\frac{1}{2}, \frac{1}{2\theta_i}\right)$.
    \item Estimate changepoints in $Z_1,\ldots,Z_n$.
    \item Fit a gamma GLM to test for a change in the rate of $Z_i$ on either side of each estimated changepoint.  
\end{enumerate}
\end{algorithm}
To carry out  Step 2 of Algorithm \ref{alg:naive}, we use the nonparametric changepoint detection method of \citet{haynes2017computationally}, implemented in the \texttt{changepoint.np} R package \citep{cpt.np}, with a BIC penalty and a minimum segment length of 10 days.

However, using the same data to estimate and test changepoints will lead to many false discoveries, as pointed out by \cite{hyun2021post} and \cite{jewell2022testing} in a related setting. 

A natural alternative is to use \emph{order-preserved sample splitting}, which involves estimating changepoints on a training set composed of odd-indexed observations, and testing those changepoints on a test set composed of even-indexed observations
\citep{10.1214/19-AOS1814}. 
Note that order-preserved sample splitting is different from Example \ref{ex:samplesplit}. Since the $Z_i$ are \emph{not} independent and identically distributed, it is \emph{not} a special case of data thinning.

\begin{algorithm}[Order-preserved sample splitting approach for changepoint detection]\label{alg:opss}
\textcolor{white}{.}
\begin{enumerate}
    \item Compute $Z_i := X_i^2$. Note that $Z_i  \sim \text{Gamma}\left(\frac{1}{2}, \frac{1}{2\theta_i}\right)$.
    \item Assume $n$ is even. Estimate changepoints in odd observations $Z_1,Z_3, \ldots,Z_{n-1}$.
    \item Fit a gamma GLM to test for a change in the rate of $Z_i$ on either side of each estimated changepoint using even observations $Z_2,Z_4, \ldots,Z_{n}$.
\end{enumerate}
\end{algorithm}

In Step 2 of Algorithm \ref{alg:opss}, we again use the \verb=changepoint.np= R package with a BIC penalty, but with a minimum segment length of five points (corresponding to 10 days).

\cite{yu2020veridical} point out that it is important for the findings of a data analysis to be stable across perturbations of the data; a similar argument underlies the  stability selection proposal of \cite{meinshausen2010stability}. We may wish to assess stability by repeating the splitting procedure many times, and comparing the estimated and rejected changepoints across different splits of the data. 
However, deterministic approaches like Algorithms \ref{alg:naive} and \ref{alg:opss} do not lend themselves to repetition.

Generalized data thinning offers a solution to this problem. Each time the procedure is run, sampling from $G_t$ produces a different pair of independent training and test sets. This allows us to assess stability of the procedure across any number of replicates.

\begin{algorithm}[Generalized data thinning approach for changepoint detection]\label{alg:gdt}
\textcolor{white}{.}
\begin{enumerate}
    \item Indirectly thin each $X_i$ through the function $S(x_i)=x_i^2$, as in Example~\ref{ex:other}.1 (with $\mu=0$).  This yields $\Xt{1}_1,\ldots,\Xt{1}_n$ and $\Xt{2}_1,\ldots,\Xt{2}_n$, where $\Xt{1}_i, \Xt{2}_i \sim \mathrm{Gamma}\left(\frac{1}{4}, \frac{1}{2 \theta_i}\right)$ and $\Xt{1}_i$ and $\Xt{2}_i$ are independent. 
    \item Estimate changepoints in $\Xt{1}_1,\ldots,\Xt{1}_n$.
    \item Fit a gamma GLM to test whether there is a change in the rate of  $\Xt{2}_1,\ldots,\Xt{2}_n$ on either side of each estimated changepoint. 
\end{enumerate}
\end{algorithm}
In Step 2 of Algorithm \ref{alg:gdt}, we again apply the nonparametric changepoint detection method, this time with the same 10-point minimum segment length used in Algorithm \ref{alg:naive}.

We first compare the methods in a simulation study; see Supplement F for  details. Figure~\ref{fig:simulation} demonstrates that in the setting where there are no true changepoints, the naive approach fails to control the type 1 error rate. By contrast, both order-preserved sample splitting and generalized data thinning  control the type 1 error rate. Figures S2~and~S3 of Supplement F overlay the simulated data with the detected changepoints, further illustrating that the naive approach routinely mistakes noise for signal.

\begin{figure}[!h]
\begin{center}
\includegraphics[width=0.8\textwidth,trim={0 3mm 0 1mm},clip]{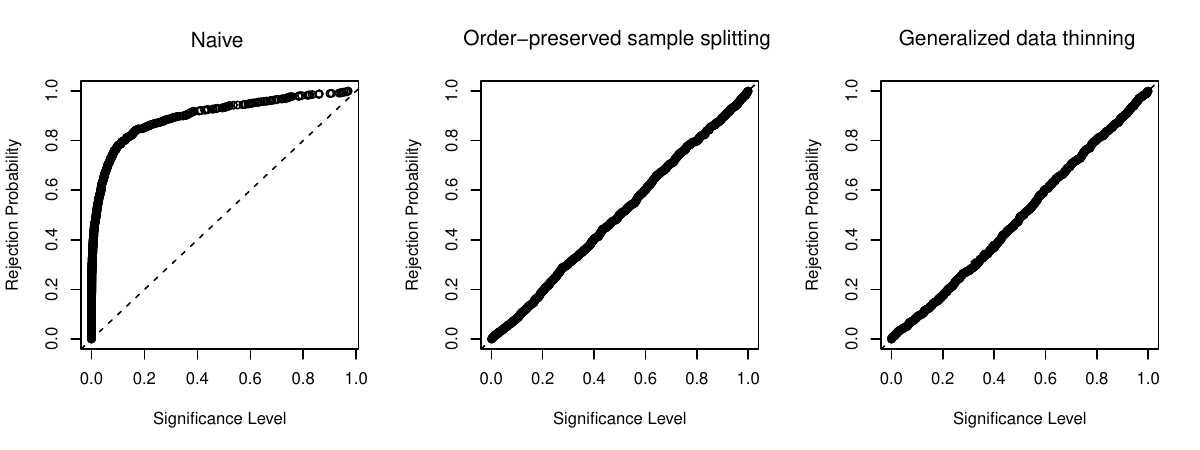}
\end{center}
\caption{Type 1 error rate of naive (Algorithm~\ref{alg:naive}), order-preserved sample splitting (Algorithm~\ref{alg:opss}), and generalized data thinning (Algorithm~\ref{alg:gdt}) approaches to testing for a change in variance, in a setting where the variance is truly constant.}
\label{fig:simulation}
\end{figure}

Turning back to the wind speed data, the top three panels of Figure \ref{fig:application}  show the results of applying the naive, order-preserved sample splitting, and generalized data thinning approaches. To account for the effects of multiple comparisons, when testing changepoints we apply a Bonferroni correction by dividing the standard 0.05 threshold by the number of detected changepoints.
We see that the naive method's p-values are below the Bonferroni corrected threshold
for over a third of the estimated changepoints. By contrast, the order-preserved sample splitting and generalized data thinning approaches give similar results with no rejections of the null hypothesis. 
In light of the results in Figure \ref{fig:simulation} and Supplement F, we believe that most of the  changepoints for which we rejected the null hypothesis using 
the naive approach are false positives. 
\begin{figure}[!h]
\begin{center}
\includegraphics[width=0.75\textwidth]{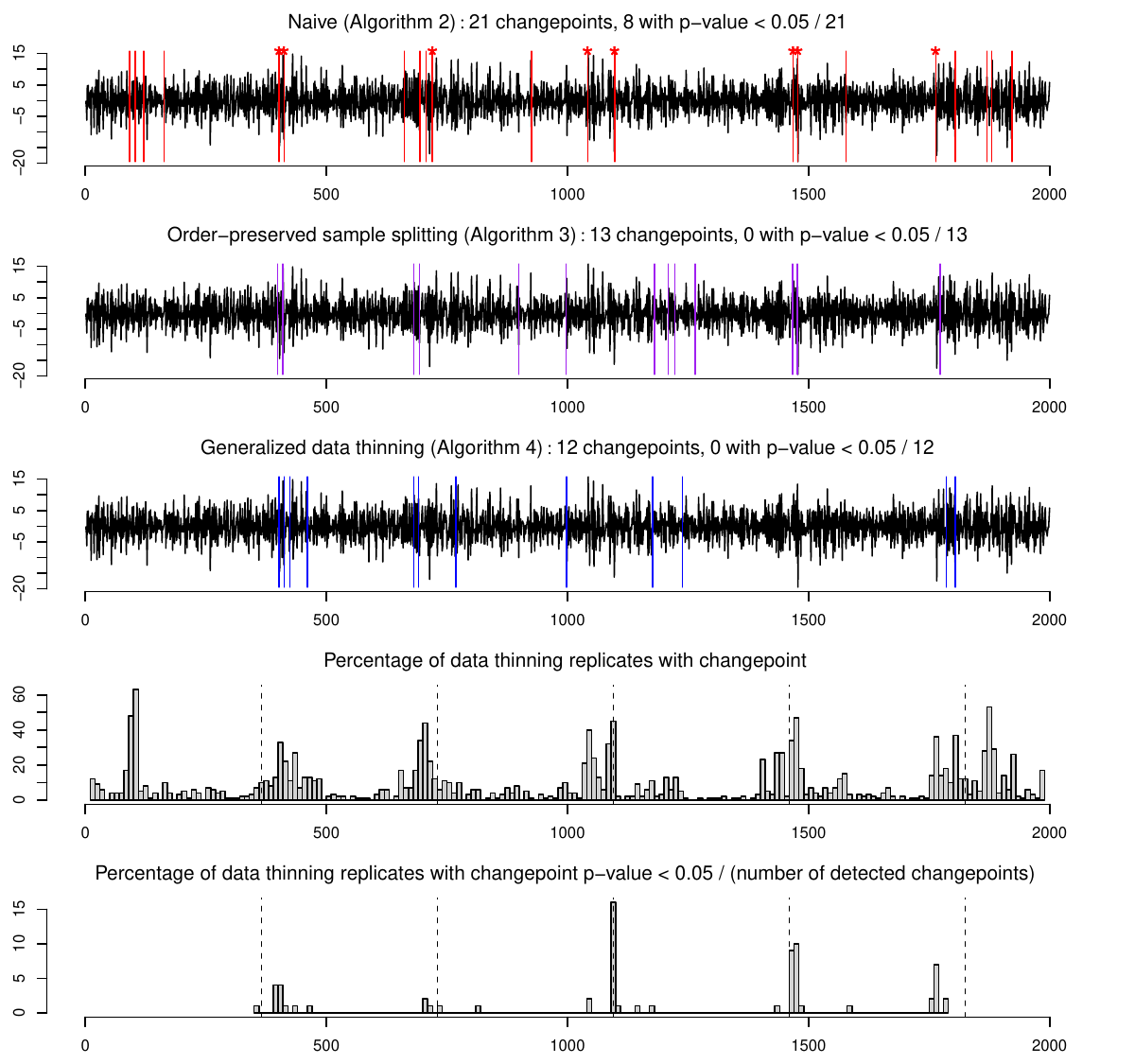}
\end{center}
\caption{Results for the wind speed analysis in Section \ref{sec:changepoint}. In each panel, the $x$-axis indexes the days. 
\emph{First three rows:} Wind speed data over time with results of each approach (Algorithms~\ref{alg:naive}, \ref{alg:opss}, and \ref{alg:gdt}) overlayed: Vertical lines indicate changepoints estimated and asterisks indicate those estimated changepoints for which the computed p-value was below 0.05 divided by the number of detected changepoints. 
\emph{Fourth row:} We binned the 2,000 days into  10-day windows. For each 10-day window, we display the percentage of replicates of the generalized data thinning approach for which at least one changepoint was estimated on the training set. Dashed lines are drawn every 365 days. \emph{Fifth row:} 
For each 10-day window, we display the percentage of replicates of the generalized data thinning approach for which at least one changepoint was estimated on the training set \emph{and} that estimated changepoint had a test set p-value below 0.05 divided by the number of detected changepoints.}
\label{fig:application}
\end{figure}

We now turn to the lower two panels of Figure~\ref{fig:application} to see the advantage of the generalized data thinning approach over the order-preserved sample splitting approach. As mentioned previously, the generalized data thinning approach is amenable to a stability analysis whereas the order-preserved sample splitting approach is not.  
In this spirit,  
we repeatedly apply Algorithm \ref{alg:gdt} a total of 100 times and compare results across replicates.
The fourth panel of  Figure \ref{fig:application} displays, for each 10-day window, the percentage of replicates in which at least one changepoint was estimated  using the training set. The fifth panel displays, for each 10-day window, the percentage of replicates for  which there was at least one changepoint estimated using the training set \emph{and} that estimated changepoint had a test set p-value below the Bonferroni corrected threshold. As none of the changepoints identified are consistently found to be significant, we are skeptical that they represent true changes in variance. Additional data are likely needed to draw a definitive conclusion.

\section{Discussion}
\label{sec:discussion}

Our generalized data thinning proposal encompasses a diverse set of existing approaches for splitting a random variable into  independent random variables, from convolution-closed data thinning \citep{neufeld2023data} to sample splitting \citep{cox1975note}. It provides a lens through which these existing approaches follow from the same simple principle --- sufficiency ---  and can be derived through the same simple recipe (Algorithm~\ref{alg:recipe}).

The principle of sufficiency is key to generalized data thinning, as it enables a sampling mechanism that does not depend on unknown parameters.  When no sufficient statistic that reduces the data is available, as in the non-parametric setting of Section~\ref{sec:sample-splitting} and the Cauchy example of Section~\ref{subsec:cauchy}, then sample splitting is still possible, provided that the observations are independent and identically distributed. Conversely, in a setting with $n=1$ or where the elements of $\mathbf{X}=(X_1,\ldots,X_n)$ are not independent and identically distributed, sample splitting may not be possible, but other  generalized thinning approaches may be available. 

For example, consider a regression setting with a fixed design, in which  each response $Y_i$ has a potentially distinct distribution determined by its corresponding feature vector $\mathbf{x}_i$, for $i=1,\ldots,n$.  It is typical to recast this as random pairs $(\mathbf{x}_1,Y_1),\ldots,(\mathbf{x}_n, Y_n)$ that are independent and identically distributed from some joint distribution, thereby justifying sample splitting.  However, this amounts to viewing the model as arising from a random design, which  may not match the reality of how the design matrix was generated, and may not be well-aligned with the goals of the data analysis.  For instance, recall the example given in the introduction: Given a dataset consisting of the  $n=50$ states of the United States, it is unrealistic to treat each state as an independent and identically distributed draw, and undesirable to perform inference only on the states that were ``held out" of training. In this example, generalized data thinning could provide a more suitable alternative to sample splitting that stays true to the fixed design model underlying the data. 

The starting place for any generalized thinning strategy---whether sample splitting or otherwise---is  the assumption that the data are drawn from a distribution belonging to a family $\mathcal P$.  
An important topic of future study is the effect of model misspecification.  In particular, if we falsely assume that $X\sim P_\theta\in\mathcal P$, 
what goes wrong?  The  random variables $\Xt{1},\ldots,\Xt{K}$ generated by thinning will still satisfy the property $X=T(\Xt{1},\ldots,\Xt{K})$; however, $\Xt{1},\ldots,\Xt{K}$ may not be independent and may no longer have the intended marginals $\Qt{1}_\theta,\ldots,\Qt{K}_\theta$.  Can we quantify the effect of the model misspecification?  I.e., if the true family is ``close" to the assumed family, will  $\Xt{1},\ldots,\Xt{K}$ be only weakly dependent, and will the marginals be close to $\Qt{1}_\theta,\ldots,\Qt{K}_\theta$? Some initial answers to these questions can be found in \cite{neufeld2023data} and \cite{rasines2021splitting}.

In the introduction, we noted that generalized data thinning with $K=2$  is a $(U,V)$-decomposition, as defined in \cite{rasines2021splitting}.  We elaborate on that connection here.  The $(U,V)$-decomposition seeks independent random variables $U=u(X,W)$ and $V=v(X,W)$ such that $U$ and $V$ are jointly sufficient for the unknowns, where $W$ is a random variable possibly depending on $X$.  Suppose we can indirectly thin $X$ through $S(\cdot)$ by $T(\cdot)$.  This means we have produced independent random variables $\Xt{1}$ and $\Xt{2}$ for which $S(X)=T(\Xt{1},\Xt{2})$.  Since $S(X)$ is sufficient for $\theta$ on the basis of $X$, this implies that $(\Xt{1},\Xt{2})$ is jointly sufficient for $\theta$.   It follows that $(\Xt{1},\Xt{2})$ is a $(U,V)$-decomposition of $X$.  It is of interest to investigate whether there are $(U,V)$-decompositions that cannot be achieved through either direct or indirect generalized data thinning.

In Section~\ref{sec:bernoulli}, we provided an example of a family for which it is impossible to perform (non-trivial) thinning.  In such situations, one may choose to drop the requirement of independence between $\Xt{1}$ and $\Xt{2}$. We expand on this extension in Supplement G.

The data thinning strategies outlined in this paper are implemented in the \texttt{datathin} R package, available at \url{https://anna-neufeld.github.io/datathin/}. Code to reproduce the simulation study and data analysis results are available at \url{https://github.com/AmeerD/gdt-experiments}.

\section*{Acknowledgments}

We thank Nicholas Irons for identifying a problem with a previous version of the proof about Bernoulli thinning
(Section~\ref{sec:bernoulli}).   
We acknowledge funding from  the following sources: 
NIH R01 EB026908, NIH R01 DA047869, ONR N00014-23-1-2589, a Simons Investigator Award in Mathematical Modeling of Living Systems, and the Keck Foundation to DW; NIH R01 GM123993 to DW and JB; a Natural Sciences and Engineering Research Council of Canada Discovery Grant to LG; and a Natural Sciences and Engineering Research Council of Canada Postgraduate Scholarships-Doctoral to AD.

\newpage

\bibliographystyle{natbib}
\bibliography{gcs}

\newpage 
\section*{Supplementary Materials}
\appendix
\setcounter{figure}{0}
\makeatletter
\renewcommand \thefigure{S\@arabic\c@figure}
\makeatother
\section{Proofs} \label{sec:appendix}

\subsection{Proof of Theorem~\ref{thm:generalized-thinning}} \label{app:pf-sufficiency-thinning}

\begin{proof}
By their construction in Definition \ref{def:thinning}, the random variables $(\Xt{1},\ldots,\Xt{K})\sim \Qt{1}_\theta\times\cdots\times\Qt{K}_\theta$ have conditional distribution
$$
(\Xt{1},\ldots,\Xt{K})|\{X=t\}\sim G_t.
$$
Furthermore, Definition~\ref{def:thinning} tells us that $X=T(\Xt{1},\ldots,\Xt{K})$. This means that
$$
(\Xt{1},\ldots,\Xt{K})|\{T(\Xt{1},\ldots,\Xt{K})=t\}\sim G_t,
$$
which establishes part (b) of the theorem.

The distribution $G_t$ in Definition~\ref{def:thinning} does not depend on $\theta$ (note that it is associated with the entire family $\mathcal P$, not a particular distribution $P_\theta$). By the definition of sufficiency, the fact that the conditional distribution $(\Xt{1},\ldots,\Xt{K})|T(\Xt{1},\ldots,\Xt{K})$ does not depend on $\theta$ implies that $T(\Xt{1},\ldots,\Xt{K})$ is sufficient for $\theta$. This proves (a).
\end{proof}

\subsection{Proof of Theorem~\ref{thm:natural-exponential-families}} \label{app:pf-exp-fam}

\begin{proof}
We start by proving the $\impliedby$ direction.
Suppose $H$ is the convolution of $H_1,\ldots, H_K$.  We follow the recipe given in Algorithm~\ref{alg:recipe}:

\begin{enumerate}
    \item We choose $\cQt{k}=\mathcal P^{H_k}$ for $k=1,\ldots,K$.
    \item Let $(\Xt{1},\ldots,\Xt{K})\sim P^{H_1}_\theta\times \cdots\times P^{H_K}_\theta$.  This joint distribution satisfies
    $$
    \prod_{k=1}^KdP^{H_k}_\theta(\xt{k})=\exp\left\{\left(\sum_{k=1}^K\xt{k}\right)^\top\theta-\sum_{k=1}^K\psi_{H_k}(\theta)\right\}\prod_{k=1}^KdH_k(\xt{k}).
    $$
    By the factorization theorem, we find that $T(\Xt{1},\ldots,\Xt{K})=\sum_{k=1}^K\Xt{k}$ is sufficient for $\theta$.
    \item It remains to determine the distribution of $U=T(\Xt{1},\ldots,\Xt{K})$.  This random variable is the convolution of $P^{H_1}_\theta\times \cdots\times P^{H_K}_\theta$, and its distribution $\mu$ is defined by the following $K$-way integral:
    \begin{align*}
d\mu(u)&=\int\cdots\int1\left\{\sum_{k=1}^K\xt{k}=u\right\}\prod_{k=1}^KdP^{H_k}_\theta(\xt{k})\\
&=\exp\left\{u^\top\theta-\sum_{k=1}^K\psi_{H_k}(\theta)\right\}\int\cdots\int1\left\{\sum_{k=1}^K\xt{k}=u\right\}\prod_{k=1}^KdH_k(\xt{k})\\
&=\exp\left\{u^\top\theta-\psi_{H}(\theta)\right\}dH(u)\\
&=dP^H_\theta(u),
    \end{align*}
where in the second-to-last equality we use the assumption that $H$ is the $K$-way convolution of $H_1,\ldots,H_K$ and the fact that the moment generating function of a convolution is the product of the individual moment generating functions (and recalling that $\psi_H$ is the logarithm of the moment generating function of $H$). This establishes that $T(\Xt{1},\ldots,\Xt{K})\sim P^H_\theta$. By Theorem~\ref{thm:generalized-thinning}, the family $\mathcal P^H$ is thinned by this choice of $T(\cdot)$.
\end{enumerate}

We now prove the $\implies$ direction.
Suppose that $\mathcal P^H$ can be $K$-way thinned into $\mathcal P^{H_1},\ldots, \mathcal P^{H_K}$ using the summation function.  Then applying Definition~\ref{def:thinning} with $\theta=0$, we can take $X\sim P^{H}_0$ and produce $(\Xt{1},\ldots,\Xt{K})\sim P^{H_1}_0\times\cdots\times P^{H_K}_0$ for which $X=\Xt{1}+\cdots+\Xt{K}$.  Noting that $P^{H_k}_0=H_k$ for all $k$ and $P^{H}_0=H$, this proves that $H$ is a $K$-way convolution of $H_1,\ldots,H_K$. 
\end{proof}

\subsection{Proof of Theorem~\ref{thm:backwards-natexpfam}}
\label{app:pf-backwards}

\begin{proof}
Suppose that $X\sim P_\theta$ is a natural exponential family with $d$-dimensional parameter $\theta$ that can be thinned by $T(\cdot)$ into $\Xt{k}\overset{ind}\sim Q_\theta^{(k)}$ for $k=1,\dots K$. By Theorem \ref{thm:generalized-thinning}, $T(\Xt{1},\dots,\Xt{K})$ is a sufficient statistic for $\theta$ on the basis of $\Xt{1},\dots,\Xt{K}$, which implies that the conditional distribution $(\Xt{1},\dots,\Xt{K})|T(\Xt{1},\dots,\Xt{K})=t$ does not depend on $\theta$. We can write the conditional density with respect to the appropriate dominating measure as
\begin{align*}
& f_{\Xt{1},\dots,\Xt{K}|T(\Xt{1},\dots,\Xt{K})=t}(\xt{1},\dots,\xt{K}) \\ 
&= \frac{f_{\Xt{1},\dots,\Xt{K}}(\xt{1},\dots,\xt{K})1\{T(\xt{1},\dots,\xt{K})=t\}}{f_{T(\Xt{1},\dots,\Xt{K})}(t)} \\
&= \frac{q^{(1)}_\theta(\xt{1})\dots q^{(K)}_\theta(\xt{K}) 1\{T(\xt{1},\dots,\xt{K})=t\}}{\exp(T(\xt{1},\dots,\xt{K})^\top\theta - \psi(\theta))h(T(\xt{1},\dots,\xt{K}))} \\
&= \frac{\prod_{k=1}^Kq^{(k)}_\theta(\xt{k})}{\exp(T(\xt{1},\dots,\xt{K})^\top\theta - \psi(\theta))} \cdot \frac{ 1\{T(\xt{1},\dots,\xt{K})=t\}}{h(T(\xt{1},\dots,\xt{K}))},
\end{align*}
where in the second equality, we used that $T(\Xt{1},\dots,\Xt{K})\overset{D}{=}X\sim P_\theta$.

As this distribution cannot depend on $\theta$, the first fraction must be constant in $\theta$. That is, for any fixed $\theta_0\in\Omega$,
\begin{align}\label{eq:qtheta}
&\frac{\prod_{k=1}^Kq^{(k)}_\theta(\xt{k})}{\exp(T(\xt{1},\dots,\xt{K})^\top\theta - \psi(\theta))} = \frac{\prod_{k=1}^Kq^{(k)}_{\theta_0}(\xt{k})}{\exp(T(\xt{1},\dots,\xt{K})^\top\theta_0 - \psi(\theta_0))}  \nonumber\\
\iff&\prod_{k=1}^K \frac{q^{(k)}_\theta(\xt{k})}{q^{(k)}_{\theta_0}(\xt{k})} = \exp(T(\xt{1},\dots,\xt{K})^\top(\theta-\theta_0) - (\psi(\theta) - \psi(\theta_0))) \nonumber\\
\iff&T(\xt{1},\dots,\xt{K})^\top(\theta-\theta_0)  = \sum_{k=1}^K \left[\log q^{(k)}_\theta(\xt{k}) + \frac{1}{K}\psi(\theta) - \log q^{(k)}_{\theta_0}(\xt{k}) - \frac{1}{K}\psi(\theta_0)\right]. 
\end{align}

To proceed, we must first confirm that the term inside the summation on the right-hand side is linear in $\theta-\theta_0$. To see this, observe that if we replace $\xt{1}$ with $\tilde{x}^{(1)}$, then 
\begin{align*}
&\left[T(\xt{1},\xt{2},\dots,\xt{K})- T(\tilde{x}^{(1)},\xt{2},\dots,\xt{K})\right]^\top(\theta-\theta_0) \\
&= \log q^{(1)}_\theta(\xt{1}) - \log q^{(1)}_{\theta_0}(\xt{1}) - \log q^{(1)}_\theta(\tilde{x}^{(1)}) + \log q^{(1)}_{\theta_0}(\tilde{x}^{(1)}) \\
&= a^{(1)}(\xt{1},\theta) - a^{(1)}(\xt{1},\theta_0) - a^{(1)}(\tilde{x}^{(1)},\theta) + a^{(1)}(\tilde{x}^{(1)},\theta_0)
\end{align*}
where $a^{(1)}(x,\theta)=\log q_\theta^{(1)}(x)$. Since the initial expression in the previous string of equalities is linear in $\theta$, the same must be true for the final expression, implying that $a^{(1)}$ must be of the form 
$$
a^{(1)}(x,\theta) = T^{(1)}(x)^\top\theta + f^{(1)}(x) + g^{(1)}(\theta) 
$$
for some functions $T^{(1)}(\cdot)$, $f^{(1)}(\cdot)$, and $g^{(1)}(\cdot)$.

Substituting into the above, 
$$
\left[T(\xt{1},\xt{2},\dots,\xt{K})- T(\tilde{x}^{(1)},\xt{2},\dots,\xt{K})\right]^\top(\theta-\theta_0)
= \left[T^{(1)}(\xt{1})-T^{(1)}(\tilde{x}^{(1)})\right]^\top(\theta-\theta_0).
$$

Applying the same logic to every $k=1,\dots,K$ in sequence, we have that for any $k$,
\begin{align*}
&\left[T(\tilde{x}^{(1)},\dots,\tilde{x}^{(k-1)},\xt{k},\xt{k+1},\dots,\xt{K})- T(\tilde{x}^{(1)},\dots,\tilde{x}^{(k-1)},\tilde{x}^{(k)},\xt{k+1},\dots,\xt{K})\right]^\top(\theta-\theta_0) \\
&= \left[T^{(k)}(\xt{k})-T^{(k)}(\tilde{x}^{(k)})\right]^\top(\theta-\theta_0)
\end{align*}
for some function $T^{(k)}(\cdot)$.

Summing over $k=1,\dots,K$ then yields
$$
\left[T(\xt{1},\dots,\xt{K})- T(\tilde{x}^{(1)},\dots,\tilde{x}^{(K)})\right]^\top(\theta-\theta_0) 
= \left[\sum_{k=1}^KT^{(k)}(\xt{k})-\sum_{k=1}^KT^{(k)}(\tilde{x}^{(k)})\right]^\top(\theta-\theta_0).
$$

Since $\mathcal{P}=\{P_\theta:\theta\in\Omega\}$ is a $d$-dimensional full-rank natural exponential family, there exists a $\theta_0\in\Omega$ and $\epsilon>0$ such that $\theta = \theta_0 + \epsilon v\in\Omega$ for every $v\in\mathbb{R}^d$ such that $\|v\|_2=1$. 

Since the previous display is true for every pair of $\theta$ and $\theta_0$, selecting pairs such that $\theta-\theta_0=\epsilon v$ simplifies the above into
$$
\left[T(\xt{1},\dots,\xt{K})- T(\tilde{x}^{(1)},\dots,\tilde{x}^{(K)})\right]^\top v 
= \left[\sum_{k=1}^KT^{(k)}(\xt{k})-\sum_{k=1}^KT^{(k)}(\tilde{x}^{(k)})\right]^\top v.
$$

As the above holds for all $v\in\mathbb{R}^d$ such that $||v||_2=1$, restricting our attention to the standard basis vectors implies that 
$$
T(\xt{1},\dots,\xt{K})- T(\tilde{x}^{(1)},\dots,\tilde{x}^{(K)})
= \sum_{k=1}^KT^{(k)}(\xt{k})-\sum_{k=1}^KT^{(k)}(\tilde{x}^{(k)})
$$

and furthermore that 
$$
T(\xt{1},\dots,\xt{K})
= \sum_{k=1}^KT^{(k)}(\xt{k})+c.
$$

Without loss of generality, $c$ can be absorbed into the $T^{(k)}(\cdot)$ functions, thus proving the claim that if a natural exponential family can be thinned, then the function $T(\cdot)$ must be a summation of the form $T(\Xt{1},\dots,\Xt{K})=\sum_{k=1}^K T^{(k)}(\Xt{k})$ for some functions $T^{(k)}(\cdot)$ for $k=1,\dots,K$.

Finally, plugging this expression into \eqref{eq:qtheta} gives
$$
\sum_{k=1}^KT^{(k)}(\xt{k})^\top(\theta-\theta_0)  = \sum_{k=1}^K \left[\log q^{(k)}_\theta(\xt{k}) + \frac{1}{K}\psi(\theta) - \log q^{(k)}_{\theta_0}(\xt{k}) - \frac{1}{K}\psi(\theta_0)\right], 
$$
which shows that the functions $q_\theta^{(k)}(\cdot)$ can be characterised as
\begin{align*}
&T^{(k)}(\xt{k})^\top (\theta-\theta_0) = \log q^{(k)}_\theta(\xt{k}) + \frac{1}{K}\psi(\theta) - \log q^{(k)}_{\theta_0}(\xt{k}) - \frac{1}{K}\psi(\theta_0) \\
\iff&\log q^{(k)}_\theta(\xt{k}) = T^{(k)}(\xt{k})^\top (\theta-\theta_0) - \frac{1}{K}\psi(\theta) + \log q^{(k)}_{\theta_0}(\xt{k}) + \frac{1}{K}\psi(\theta_0) \\
\iff& q^{(k)}_\theta(\xt{k}) =  q^{(k)}_{\theta_0}(\xt{k})\exp\left(T^{(k)}(\xt{k})^\top (\theta-\theta_0) - \frac{1}{K}\psi(\theta) + \frac{1}{K}\psi(\theta_0)\right).
\end{align*}

Thus, $q_\theta^{(k)}(\cdot)$ is the density of an exponential family with sufficient statistic $T^{(k)}(\cdot)$ and carrier density given by $h^{(k)}(\xt{k})\propto q^{(k)}_{\theta_0}(\xt{k})\exp(-T^{(k)}(\xt{k})^\top \theta_0)$.

\end{proof}

\subsection{Proof of Proposition~\ref{prop:thin-S}} \label{app:pf-thin-S}

\begin{proof}
The result follows from a chain of equalities:
\begin{align*}
    I_X(\theta)&=I_{S(X)}(\theta)\\
    &=I_{T(\Xt{1},\ldots,\Xt{K})}(\theta)\\
    &=I_{(\Xt{1},\ldots,\Xt{K})}(\theta)\\
    &=\sum_{k=1}^K I_{\Xt{k}}(\theta).
\end{align*}
The first equality is true because $S(X)$ is sufficient for $\theta$ based on $X$.  The second equality follows from the definition of thinning $S(X)$ into $\Xt{1},\ldots,\Xt{K}$ using $T(\cdot)$.  The third equality follows from Theorem~\ref{thm:generalized-thinning}, which tells us that $T(\Xt{1},\ldots,\Xt{K})$ is sufficient for $\theta$ based on $(\Xt{1},\ldots,\Xt{K})$. The final equality follows from independence.
\end{proof}

\subsection{Proof of Theorem~\ref{thm:bernoulli}} \label{app:bernoulli-convolution}

\begin{proof}
We begin by providing some intuition. Since $Z^{(1)}$ and $Z^{(2)}$ are non-constant random variables, their supports each must contain more than one element. Therefore, by independence, the support of $Z^{(1)} + Z^{(2)}$ must contain more than two elements and thus cannot be Bernoulli.

Formally, let $\Qt{1}$ and $\Qt{2}$ be the distributions of $Z^{(1)}$ and $Z^{(2)}$, respectively.
If $Z^{(1)}+Z^{(2)}$ were Bernoulli, then
\begin{align*}
    1&=\mathbb P(Z^{(1)}+Z^{(2)}\in\{0,1\})\\
    &=\int\mathbb P(Z^{(2)}\in\{0-z^{(1)},~1-z^{(1)}\}|Z^{(1)}=z^{(1)})d\Qt{1}(z^{(1)})\\
    &=\int\mathbb P(Z^{(2)}\in\{0-z^{(1)},~1-z^{(1)}\})d\Qt{1}(z^{(1)}),
\end{align*}
where the last equality follows by independence of $Z^{(1)}$ and $Z^{(2)}$.  
For this integral to equal 1, we would need
\begin{equation}\label{eq:prob1}
    \mathbb P(Z^{(2)}\in\{0-z^{(1)},~1-z^{(1)}\})=1 \text{ for $\Qt{1}$-almost every $z^{(1)}$}
\end{equation}
since $P(Z^{(2)}\in\{0-z^{(1)},~1-z^{(1)}\})$ is bounded above by $1$. For $Z^{(1)}$ to be non-constant, there must be at least two distinct points $a$ and $b$ such that \eqref{eq:prob1} holds with $z^{(1)}=a$ and holds with $z^{(1)}=b$.  Since the intersection of two probability 1 sets is a set that holds with probability 1, we have that 
$$
\mathbb P(Z^{(2)}\in\{-a,~1-a\}\cap\{-b,~1-b\})=1,
$$
from which it follows that $\{-a,~1-a\}\cap\{-b,~1-b\}$ is non-empty.  However, there is no choice of $a\neq b$ for which this intersection has more than one element (which is required for $Z^{(2)}$ to be non-constant).  Thus we arrive at a contradiction.

\end{proof}

\subsection{Proof of Corollary~\ref{cor:bern2}} \label{app:bern2}

\begin{proof}
Since the Bernoulli family is a natural exponential family, if at least one of the conclusions of Theorem $\ref{thm:backwards-natexpfam}$ is always false for the Bernoulli, then the contrapositive of Theorem $\ref{thm:backwards-natexpfam}$ will imply that the Bernoulli distribution cannot be thinned by any function $T(\cdot)$. 

Suppose that $X\sim\text{Bernoulli}(\theta)$. Consider the first conclusion, namely that the thinning function $T(\xt{1},\dots,\xt{K})$ is of the form $\sum_{k=1}^KT^{(k)}(\xt{k})$. This would imply that $\sum_{k=1}^KT^{(k)}(\Xt{k})=X\sim\text{Bernoulli}(\theta)$.  However, by Theorem \ref{thm:bernoulli}, $T(\cdot)$ cannot be a convolution of independent, non-constant random variables. Therefore, the second conclusion can only be true if $\Xt{1},\dots,\Xt{K}$ are not mutually independent, some or all of $\Xt{1},\dots,\Xt{K}$ are constant, or some or all of $T^{(1)}(\cdot),\dots,T^{(K)}(\cdot)$ are constant functions. All three cases violate the second conclusion of Theorem \ref{thm:backwards-natexpfam} that $\Xt{k}$ are independent exponential families. Therefore, both conclusions of Theorem \ref{thm:backwards-natexpfam} cannot be simultaneously true, thus proving the claim. 
\end{proof}

\section{Connecting Example~\ref{exmp:gaussian} to prior work} \label{app:UV}
 
Other authors have considered obtaining two independent Gaussian random variables $\mathbf{U}$ and $\mathbf{V}$ from a single Gaussian random variable $\mathbf{X} \sim N_n(\boldsymbol\theta, \mathbf{I}_n)$ by generating $\mathbf{W} \sim N_n(\mathbf{0}_n, \gamma \mathbf{I}_n)$ for a tuning parameter $\gamma > 0$, and then setting $\mathbf{U} = \mathbf{X} + \mathbf{W}$ and  $\mathbf{V} = \mathbf{X} - \gamma^{-1}\mathbf{W}$. Then, $\mathbf{U} \sim N_n(\boldsymbol\theta, (1 + \gamma) \mathbf{I}_n)$ and $\mathbf{V} \sim N_n(\boldsymbol\theta, (1 + \gamma^{-1}) \mathbf{I}_n)$ are independent.
\citet{rasines2021splitting} and \citet{leiner2022data} applied this decomposition to address Scenario 1 in Section \ref{sec:introduction}. Additionally, \citet{rasines2021splitting} showed that this leads to asymptotically valid inference under certain regularity conditions, even when $\mathbf{X}$ is not normally distributed.  \citet{tian2020prediction} and \citet{oliveira2021unbiased} applied this decomposition to address Scenario 2 in Section \ref{sec:introduction}. 

This decomposition is in fact identical to Example~\ref{exmp:gaussian} up to scaling, with $\mathbf{X}^{(1)}=\epsilon_1\mathbf{U}$, $\mathbf{X}^{(2)}=\epsilon_2\mathbf{V}$, $\epsilon_1=(1+\gamma)^{-1}$, and $\epsilon_2=1-\epsilon_1$.
In particular, to thin $\mathbf{X} \sim N_n(\boldsymbol\theta, \mathbf{I}_n)$ by addition into $(\mathbf{X}^{(1)},\mathbf{X}^{(2)})$, we sample from $G_{\mathbf{X}}$ where $G_{\mathbf{t}}$, defined in Theorem~\ref{thm:generalized-thinning}, can be shown to equal the singular multivariate normal distribution
$$
N_{2n}\left(\begin{pmatrix}\epsilon_1\mathbf{t}\\ \epsilon_2\mathbf{t}\end{pmatrix},\epsilon_1\epsilon_2\begin{pmatrix}\mathbf{I}_n&-\mathbf{I}_n\\-\mathbf{I}_n&\mathbf{I}_n\end{pmatrix}\right).
$$

Sampling from $G_{\mathbf{X}}$ is equivalent to sampling $\mathbf{W} \sim N_n(\mathbf{0}_n, \gamma \mathbf{I}_n)$ (independent of $\mathbf{X}$) and then generating $\mathbf{X}^{(1)} = \epsilon_1 (\mathbf{X} + \mathbf{W})$ and $\mathbf{X}^{(2)} = \epsilon_2 (\mathbf{X} - \gamma^{-1} \mathbf{W})$.
These ideas can easily be generalized to thin $\mathbf{X} \sim N_n(\boldsymbol\theta, \boldsymbol\Sigma)$ with a known positive definite covariance matrix $\boldsymbol\Sigma$.

\section{Derivations of thinning procedures} \label{sec:derivations}

\subsection{Exponential families}

\subsubsection{Weibull distribution}
\label{subsec:weibull-proofs}

The gamma family with known shape $\alpha$ and unknown rate $\theta$ admits a collection of thinning functions, indexed by a hyperparameter $\nu>0$, that thin the gamma family into the Weibull family.

\begin{exmp}[Thinning $\text{Gamma}(\alpha,\theta)$ with $\alpha=K$ known, approach 3]
\label{ex:weibull}
Recall that the Weibull distribution with known shape parameter $\nu$ (but varying scale $\lambda$) is a general exponential family.  Then, starting with $\Xt{k} \overset{iid}{\sim} \text{Weibull}(\lambda, \nu)$ for $k=1,\dots,K$, we have that $T^{(k)}(\xt{k}) = (\xt{k})^\nu$. We can thus apply Proposition~\ref{prop:exp-family} using the function
$$
T(\xt{1},\dots,\xt{K})=\sum_{k=1}^K(\xt{k})^\nu
$$
to thin the distribution of $\sum_{k=1}^K(\Xt{k})^\nu$ into $(\Xt{1},\dots,\Xt{K})$. 
As $\sum_{k=1}^K(\Xt{k})^\nu\sim\text{Gamma}(K,\lambda^{-\nu})$, 
taking $\lambda=\theta^{-1/\nu}$ yields the desired result.  

To generate $(\Xt{1},\dots,\Xt{K})$, we can first apply the $K$-fold gamma thinning result discussed in Example \ref{ex:dtgamma} with $\epsilon_k = \frac{1}{K}$ to generate $Y^{(k)} \overset{iid}{\sim} \text{Exp}(\lambda^{-\nu})$, and then compute $\Xt{k} = (Y^{(k)})^{\frac{1}{\nu}}$. 
\end{exmp}

\begin{proof}[Proof of Example~\ref{ex:weibull}]
We must prove that if $\Xt{k} \overset{iid}{\sim} \text{Weibull}(\lambda, \nu)$, for $k=1,\dots,K$, then $\sum_{k=1}^K\left(\Xt{k}\right)^\nu\sim \text{Gamma}(K, \lambda^{-\nu})$.

Recalling that the gamma distribution is convolution-closed in its shape parameter, it is sufficient to show for a single $\Xt{k}\sim\text{Weibull}(\lambda, \nu)$ random variable that $(\Xt{k})^\nu \sim \text{Gamma}(1,\lambda^{-\nu})=\text{Exp}(\lambda^{-\nu})$, where $\nu>0$. Denote $Z=(\Xt{k})^\nu$. Then,

\begin{align*}
    f_Z(z) &= f_{\Xt{k}}\left(z^{\frac{1}{\nu}}\right)\left|\frac{d \xt{k}}{d z}\right| \\
    &\propto \left(z^{\frac{1}{\nu}}\right)^{\nu-1} \exp\left(-\left(\frac{z^{\frac{1}{\nu}}}{\lambda}\right)^\nu\right)\left|\frac{1}{\nu}z^{-\frac{\nu-1}{\nu}}\right| \\
    &\propto \exp\left(-\lambda^{-\nu}z\right).
\end{align*}

The above implies that $(\Xt{k})^\nu \sim \text{Exp}(\lambda^{-\nu})$, and thus $\sum_{k=1}^K\left(\Xt{k}\right)^\nu\sim \text{Gamma}(K, \lambda^{-\nu})$ as required.
    
\end{proof}

\subsubsection{Beta distribution}
\label{subsec:beta-proofs}

\begin{proof}[Proof of Example~\ref{ex:beta}]
We must prove the following three claims:
\begin{enumerate}
    \item \emph{If $\Xt{k} \sim \text{Beta}\left(\frac{1}{K}\theta + \frac{k-1}{K}, \frac{1}{K}\beta\right)$, for $k=1,\dots,K$, and $\Xt{k}$ are mutually independent, then $\left(\prod_{k=1}^K \Xt{k}\right)^{\frac{1}{K}} \sim \text{Beta}(\theta, \beta)$.}

Recall that the beta distribution is fully characterised by its moments due to its finite support \citep{feller1971}. Also recall that the expectation of the $r$th power of a $X\sim\text{Beta}(\theta,\beta)$ random variable is $E[X^r]=\frac{B(\theta+r,\beta)}{B(\theta,\beta)}$ where $B$ is the beta function. Finally, note that the Gauss multiplication theorem \citep[page 256]{abramowitz1972handbook} says that
$$
\prod_{k=1}^K\Gamma\left(z+\frac{k-1}{K}\right)=(2\pi)^{\frac{K-1}{2}}K^{\frac{1}{2}-Kz}\Gamma(Kz).
$$

Then, the $r$th moment of $\left(\prod_{k=1}^K \Xt{k}\right)^{\frac{1}{K}}$ is

\begin{align*} 
E\left[\left(\left(\prod_{k=1}^K \Xt{k}\right)^{\frac{1}{K}}\right)^r\right] &= \prod_{k=1}^KE\left[ \left(\Xt{k}\right)^{\frac{r}{K}}\right] \\
&= \prod_{k=1}^K\frac{B(\frac{1}{K}\theta+\frac{k-1}{K}+\frac{r}{K},\frac{1}{K}\beta)}{B(\frac{1}{K}\theta+\frac{k-1}{K},\frac{1}{K}\beta)} \\
&= \prod_{k=1}^K\frac{\frac{\Gamma(\frac{1}{K}\theta+\frac{k-1}{K}+ \frac{r}{K})\Gamma(\frac{1}{K}\beta)}{\Gamma(\frac{1}{K}\theta +\frac{k-1}{K}+ \frac{r}{K}+\frac{1}{K}\beta)}}{\frac{\Gamma(\frac{1}{K}\theta+\frac{k-1}{K})\Gamma(\frac{1}{K}\beta)}{\Gamma(\frac{1}{K}\theta+\frac{k-1}{K}+\frac{1}{K}\beta)}} \\
&= \prod_{k=0}^{K-1}\frac{\Gamma(\frac{1}{K}\theta+ \frac{r}{K}+\frac{k}{K})\Gamma(\frac{1}{K}\theta+\frac{1}{K}\beta+\frac{k}{K})}{\Gamma(\frac{1}{K}\theta+\frac{1}{K}\beta+ \frac{r}{K}+\frac{k}{K})\Gamma(\frac{1}{K}\theta+\frac{k}{K})} \\
&= \frac{\left[\prod_{k=0}^{K-1}\Gamma(\frac{1}{K}\theta+ \frac{r}{K}+\frac{k}{K})\right]\left[\prod_{k=0}^{K-1}\Gamma(\frac{1}{K}\theta+\frac{1}{K}\beta+\frac{k}{K})\right]}{\left[\prod_{k=0}^{K-1}\Gamma(\frac{1}{K}\theta+\frac{1}{K}\beta+ \frac{r}{K}+\frac{k}{K})\right]\left[\prod_{k=0}^{K-1}\Gamma(\frac{1}{K}\theta+\frac{k}{K})\right]} \\
&= \frac{\left[(2\pi)^{\frac{K-1}{2}}K^{\frac{1}{2}-(\theta+r)}\Gamma(\theta+r)\right]\left[(2\pi)^{\frac{K-1}{2}}K^{\frac{1}{2}-(\theta+\beta)}\Gamma(\theta+\beta)\right]}{\left[(2\pi)^{\frac{K-1}{2}}K^{\frac{1}{2}-(\theta+\beta+r)}\Gamma(\theta+\beta+r)\right]\left[(2\pi)^{\frac{K-1}{2}}K^{\frac{1}{2}-\theta}\Gamma(\theta)\right]} \\
&= \frac{\Gamma(\theta+r)\Gamma(\beta)\Gamma(\theta+\beta)}{\Gamma(\theta+\beta+r)\Gamma(\theta)\Gamma(\beta)} \\
&= \frac{B(\theta+r,\beta)}{B(\theta,\beta)}.
\end{align*}
This matches the moments of a $\text{Beta}(\theta, \beta)$ distribution, implying that $\left(\prod_{k=1}^K \Xt{k}\right)^{\frac{1}{K}} \sim \text{Beta}(\theta, \beta)$ as required. 
    
    \item \emph{A sufficient statistic for $\theta$ in the joint distribution of $\Xt{1},\dots,\Xt{K}$ is 
    $$T(\Xt{1},\dots,\Xt{K})=\left(\prod_{k=1}^K \Xt{k}\right)^{\frac1{K}}.$$}

By the mutual independence of $\Xt{k}$, the joint density of $\Xt{1},\dots,\Xt{K}$ can be written as

\begin{align*} 
f_{\Xt{1},\dots,\Xt{K}}(\xt{1},\dots,\xt{K}) &= \prod_{k=1}^K f_{\Xt{k}}(\xt{k}) \\
&\propto \prod_{k=1}^K \left(\xt{k}\right)^{\frac{1}{K}\theta + \frac{k-1}{K} - 1}\left(1-\xt{k}\right)^{\frac{1}{K}\beta-1} \\
&=\left[\left(\prod_{k=1}^K \xt{k}\right)^{\frac{1}{K}}\right]^{\theta}\left[\prod_{k=1}^K \left(\xt{k}\right)^{\frac{k-1}{K}-1}\right]\left[\prod_{k=1}^K\left(1-\xt{k}\right)\right]^{\frac{1}{K}\beta-1}.
\end{align*}

    By the factorization theorem, $T(\Xt{1},\dots,\Xt{K})=\left(\prod_{k=1}^K \Xt{k}\right)^{\frac{1}{K}}$ is a sufficient statistic for $\theta$.

    \item \emph{To sample from $G_t$, i.e., the distribution of $(\Xt{1},\dots,\Xt{K})|T(\Xt{1},\dots,\Xt{K})=t$ 
    , we first sample from $(\Xt{1},\dots,\Xt{K-1})|T(\Xt{1},\dots,\Xt{K})=t$ and then recover $\Xt{K}$.}

We will show that the conditional density $f_{\Xt{1},\dots,\Xt{K-1}|T(\Xt{1},\dots,\Xt{K})=t}(\xt{1},\dots,\xt{K-1})$ is, up to a normalizing constant involving $t$, 
$$
\left[\prod_{k=1}^{K-1}\left(\xt{k}\right)^{\frac{k-K}{K}-1}\right]\left[\left(\prod_{k=1}^{K-1}\left(1-\xt{k}\right)\right)\left(1-\frac{t^K}{\prod_{k=1}^{K-1}\xt{k}}\right)\right]^{\frac{1}{K}\beta-1}.$$
We derive this as follows (where any factors not involving $\xt{1},\ldots,\xt{K-1}$ are omitted, and we write $\theta_k=\frac{\theta}{K}+\frac{k-1}{K}$):

\begin{align*}
&\quad f_{\Xt{1},\dots,\Xt{K-1}|T(\Xt{1},\dots,\Xt{K})=t}(\xt{1},\dots,\xt{K-1}) \\ 
&\propto f_{\Xt{1},\dots,\Xt{K-1},T(\Xt{1},\dots,\Xt{K})}(\xt{1},\dots,\xt{K-1},t) \\
&= f_{\Xt{1},\dots,\Xt{K}}\left(\xt{1},\dots,\xt{K-1},\frac{t^K}{\prod_{k=1}^{K-1}\xt{k}}\right)\left|\frac{\partial}{\partial t}\frac{t^K}{\prod_{k=1}^{K-1}\xt{k}}\right| \\
&=\left(\prod_{k=1}^{K-1}f_{\Xt{k}}\left(\xt{k}\right)\right)f_{\Xt{K}}\left(\frac{t^K}{\prod_{k=1}^{K-1}\xt{k}}\right)\left|\frac{Kt^{K-1}}{\prod_{k=1}^{K-1}\xt{k}}\right| \\
&\propto \left(\prod_{k=1}^{K-1}\left(\xt{k}\right)^{\theta_k-1}\left(1-\xt{k}\right)^{\frac{1}{K}\beta-1}\right) \left(\frac{t^K}{\prod_{k=1}^{K-1}\xt{k}}\right)^{\theta_K-1}\left(1-\frac{t^K}{\prod_{k=1}^{K-1}\xt{k}}\right)^{\frac{1}{K}\beta-1}\frac{1}{\prod_{k=1}^{K-1}\xt{k}} \\
&\propto \left[\prod_{k=1}^{K-1}\left(\xt{k}\right)^{\theta_k-\theta_K-1}\right]\left[\left(\prod_{k=1}^{K-1}\left(1-\xt{k}\right)\right)\left(1-\frac{t^K}{\prod_{k=1}^{K-1}\xt{k}}\right)\right]^{\frac{1}{K}\beta-1}.
\end{align*}
It remains to note that $\theta_k-\theta_K=\frac{k-K}{K}$.

To generate $(\Xt{1},\dots,\Xt{K})$, first sample from $(\Xt{1},\dots,\Xt{K-1})|T(\Xt{1},\dots,\Xt{K})=t$ with numerical sampling methods. In this example, a Metropolis algorithm with a uniform proposal over $[t^{K},1)^K$ is effective \citep{metropolis1953equation}. Then, compute $\Xt{K}=\frac{t^K}{\prod_{k=1}^{K-1}\Xt{k}}$. 
\end{enumerate}
\end{proof}

\subsubsection{Dirichlet distribution}
\label{subsec:dirichlet}

The Dirichlet distribution (which subsumes the beta distribution) on the $K$-simplex is typically parameterized by a $K$-dimensional vector $\boldsymbol\alpha$. It can also be parameterized by the mean, defined as $\boldsymbol\alpha/\sum_{k=1}^K\alpha_k$, and precision, defined as $\sum_{k=1}^K \alpha_k$. Using the mean-precision parameterization, the Dirichlet distribution with known precision $\phi$ and unknown mean $\boldsymbol\theta$ can be thinned into $K$ gamma random variables.

\begin{exmp}[Thinning $\text{Dirichlet}_K(\boldsymbol\theta,\phi)$ with $\phi$ known]
\label{ex:dirichlet}
Following the steps of Algorithm \ref{alg:recipe}, start with $K$ mutually independent gamma random variables, $\Xt{k}\sim\text{Gamma}(\theta_k\phi,\nu)$ for $k=1,\dots,K$ where $\nu>0$ is a tuning parameter chosen by the user. Then, note that   
$$
T(\Xt{1},\dots,\Xt{K}) = (\Xt{1},\dots,\Xt{K})^\top/\sum_{k=1}^K\Xt{K}
$$
is a sufficient statistic for $\boldsymbol\theta$ on the basis of $(\Xt{1},\dots,\Xt{K})$. Since $T(\Xt{1},\dots,\Xt{K}) \sim \text{Dirichlet}_K(\boldsymbol\theta,\phi)$, we can thus thin the Dirichlet distribution into $(\Xt{1},\dots,\Xt{K})$. 

To generate $(\Xt{1},\dots,\Xt{K})$, first sample from the conditional distribution of $\Xt{1}$ given $T(\Xt{1},\dots,\Xt{K})=\mathbf{t}$. This follows a $\text{Gamma}(\phi,\nu/t_1)$ distribution. Then for $k=2,\dots,K$, set $\Xt{k}=\Xt{1}t_k/t_1$.
\end{exmp}

\begin{proof}[Proof of Example~\ref{ex:dirichlet}]
We must prove that $(\Xt{1},\dots,\Xt{K})^\top/\sum_{k=1}^K\Xt{K}$ is a sufficient statistic for $\boldsymbol\theta$ in the joint distribution of $\Xt{k}\sim\text{Gamma}(\theta_k\phi,\nu)$ for $k=1,\dots,K$.

Consider the joint density of $\Xt{k}$ for $k=1,\dots,K$:

\begin{align*}
f_{\Xt{1},\dots,\Xt{K}}(\xt{1},\dots,\xt{K}) &= \prod_{k=1}^K f_{\Xt{k}}(\xt{k}) \\
&\propto \left[\prod_{k=1}^K \left(\xt{k}\right)^{\theta_k\phi-1}\exp\left(-\nu\xt{k}\right)\right] \\
&= \prod_{k=1}^K \left[\left(\xt{k}\right)^{\theta_k\phi-1}\left(\sum_{k'=1}^K\xt{k'}\right)^{(\theta_k\phi-1)-(\theta_k\phi-1)}\exp\left(-\nu\xt{k}\right)\right] \\
&= \prod_{k=1}^K \left[\left(\frac{\xt{k}}{\sum_{k'=1}^K\xt{k'}}\right)^{\theta_k\phi-1}\left(\sum_{k'=1}^K\xt{k'}\right)^{\theta_k\phi-1}\exp\left(-\nu\xt{k}\right)\right] \\
&= \left(\sum_{k'=1}^K\xt{k'}\right)^{\sum_{k=1}^K(\theta_k\phi-1)}\prod_{k=1}^K \left[\left(\frac{\xt{k}}{\sum_{k'=1}^K\xt{k'}}\right)^{\theta_k\phi-1}\exp\left(-\nu\xt{k}\right)\right] \\
&= \left(\sum_{k'=1}^K\xt{k'}\right)^{\phi\sum_{k=1}^K\theta_k-K}\prod_{k=1}^K \left[\left(\frac{\xt{k}}{\sum_{k'=1}^K\xt{k'}}\right)^{\theta_k\phi-1}\exp\left(-\nu\xt{k}\right)\right] \\
&= \left(\sum_{k'=1}^K\xt{k'}\right)^{\phi-K}\prod_{k=1}^K \left[\left(\frac{\xt{k}}{\sum_{k'=1}^K\xt{k'}}\right)^{\theta_k\phi-1}\exp\left(-\nu\xt{k}\right)\right] \\
\end{align*}

The above implies by the factorization theorem that $(\Xt{1},\dots,\Xt{K})^\top/\sum_{k=1}^K\Xt{K}$ is a sufficient statistic for $\boldsymbol\theta$ on the basis of $(\Xt{1},\dots,\Xt{K})$ as required.
    
\end{proof}

\subsubsection{Gamma distribution}
\label{subsec:gamma-proofs}

In proving Example \ref{ex:gamma1}, rather than work with gamma random variables directly, we will find it convenient to work with the logarithm of gamma random variables. We start by deriving the moment generating function of this distribution.
\begin{lemma}\label{lemma:log-gamma}
Consider a random variable $Y$ such that $e^Y\sim \text{Gamma}(\theta, \beta)$.  Then the moment generating function of $Y$ exists in a neighborhood around 0 and is given by 
$$
\Phi_{Y}(t)=\frac{\Gamma(\theta+t)}{\Gamma(\theta)\beta^{t}}.
$$
\end{lemma}
\begin{proof}[Proof of Lemma~\ref{lemma:log-gamma}]
The density of $Y$ is given by
\begin{align*}
    f_Y(y)&=f_{\text{Gamma}(\theta, \beta)}(e^y)e^y\\
    &=\frac{\beta^\theta}{\Gamma(\theta)}e^{y(\theta-1)}e^{-\beta e^y}e^y\\
    &=\frac{\beta^\theta}{\Gamma(\theta)}e^{y\theta-\beta e^y},
\end{align*}
where the extra factor of $e^y$ in the first equality is the Jacobian of the transformation.  For $t>-\theta$, the moment generating function for this random variable is given by
\begin{align*}
\Phi_Y(t)&=\mathbb E[e^{tY}]\\
    &=\int e^{ty}\frac{\beta^\theta}{\Gamma(\theta)}e^{y\theta-\beta e^y}dy\\
    &=\frac{\beta^\theta}{\Gamma(\theta)}\frac{\Gamma(\theta+t)}{\beta^{\theta+t}}\int\frac{\beta^{\theta+t}}{\Gamma(\theta+t)} e^{y(\theta+t)-\beta e^y}dy\\
    &=\frac{\Gamma(\theta+t)}{\Gamma(\theta)\beta^{t}}\int\frac{\beta^{\theta+t}}{\Gamma(\theta+t)} e^{y(\theta+t)-\beta e^y}dy\\
    &=\frac{\Gamma(\theta+t)}{\Gamma(\theta)\beta^{t}}.
    \end{align*}
The assumption that $t>-\theta$ ensures that $\theta+t>0$ so that the integrand in the second-to-last line is the density of the logarithm of a $\text{Gamma}(\theta+t,\beta)$ random variable.  Since $\theta>0$, we have established that the moment generating function exists in a neighborhood around 0.
\end{proof}

\begin{proof}[Proof of Example~\ref{ex:gamma1}]
We must prove the following three claims:
\begin{enumerate}
    \item \emph{If for $k=1,\dots,K$, $\Xt{k} \sim \text{Gamma}\left(\frac{1}{K}\theta + \frac{k-1}{K}, \frac{1}{K}\beta\right)$ and $\Xt{k}$ are mutually independent, then $\left(\prod_{k=1}^K \Xt{k}\right)^{\frac{1}{K}} \sim \text{Gamma}(\theta, \beta)$.} 

Defining $Y^{(k)}=\log\Xt{k}$ and $\bar Y_K=\frac{1}{K}\sum_{k=1}^KY^{(k)}$, observe that
$$
\left(\prod_{k=1}^K \Xt{k}\right)^{\frac{1}{K}}=e^{\bar Y_K}.
$$
Thus, our goal is to prove that $e^{\bar Y_K}\sim \text{Gamma}(\theta, \beta)$.  Since the moment generating function completely characterizes a distribution, it is sufficient to show that 
$\Phi_{\bar Y_K}(t)$ matches the expression in Lemma~\ref{lemma:log-gamma}. Applying Lemma~\ref{lemma:log-gamma} to $e^{Y^{(k)}}\sim\text{Gamma}(\theta_k,\beta_k)$, where $\theta_k=\theta/K+(k-1)/K$ and $\beta_k=\beta/K$, implies that 
$$
\Phi_{Y^{(k)}}=\frac{\Gamma(\theta_k+t)}{\Gamma(\theta_k)\beta_k^{t}}.
$$
By independence of $Y^{(1)},\ldots,Y^{(K)}$ and standard properties of the moment generating function,
\begin{align*}    
\Phi_{\bar Y_K}(t)&=\prod_{k=1}^K\Phi_{Y_K/K}(t)\\
&=\prod_{k=1}^K\Phi_{Y_K}(t/K)\\
&=\prod_{k=1}^K\frac{\Gamma(\theta_k+t/K)}{\Gamma(\theta_k)\beta_k^{t/K}}.
\end{align*}

Recalling the form of $\theta_k$ and $\beta_k$ and applying the Gauss multiplication theorem \citep[page 256]{abramowitz1972handbook} to both the numerator and denominator gives
$$
\Phi_{\bar Y_K}(t)=\frac{K^{-(\theta+t)}\Gamma(\theta+t)}{K^{-\theta}\Gamma(\theta)}\frac{1}{(\beta/K)^t}=\frac{\Gamma(\theta+t)}{\Gamma(\theta)\beta^t}.
$$
This completes the proof.

    \item \emph{A sufficient statistic for $\theta$ in the joint distribution of $\Xt{1},\dots,\Xt{K}$ is $T(\Xt{1},\dots,\Xt{K})=\left(\prod_{k=1}^K \Xt{k}\right)^{\frac{1}{K}}$.}

By the mutual independence of $\Xt{k}$, the joint density of $\Xt{1},\dots,\Xt{K}$ can be written as,

\begin{align*} 
f_{\Xt{1},\dots,\Xt{K}}(\xt{1},\dots,\xt{K}) &= \prod_{k=1}^K f_{\Xt{k}}(\xt{k}) \\
&\propto \prod_{k=1}^K \left(\xt{k}\right)^{\frac{1}{K}\theta + \frac{k-1}{K} - 1}\exp\left(-\frac{\beta}{K}\xt{k}\right) \\
&=\left[\left(\prod_{k=1}^K \xt{k}\right)^{\frac{1}{K}}\right]^{\theta}\left[\prod_{k=1}^K \left(\xt{k}\right)^{\frac{k-1}{K}-1}\right]\exp\left(-\frac{\beta}{K}\sum_{k=1}^K\xt{k}\right)
\end{align*}

    By the factorization theorem, $T(\Xt{1},\dots,\Xt{K})=\left(\prod_{k=1}^K \Xt{k}\right)^{\frac{1}{K}}$ is a sufficient statistic for $\theta$ as required.

    \item \emph{To sample from $G_t$, the conditional distribution $(\Xt{1},\dots,\Xt{K})|T(\Xt{1},\dots,\Xt{K})=t$
    , we first sample from $(\Xt{1},\dots,\Xt{K-1})|T(\Xt{1},\dots,\Xt{K})=t$ and then recover $\Xt{K}$.}

The conditional density $f_{\Xt{1},\dots,\Xt{K-1}|T(\Xt{1},\dots,\Xt{K})=t}(\xt{1},\dots,\xt{K-1})$, up to a normalizing constant depending on $t$, is
$$
\left[\prod_{k=1}^{K-1}\left(\xt{k}\right)^{\frac{k-K}{K}-1}\right]\exp\left(-\frac{\beta}{K}\left(\sum_{k=1}^{K-1}\xt{k} + \frac{t^K}{\prod_{k=1}^{K-1}\xt{k}}\right)\right).
$$
The derivation is as follows (where any factors not involving $\xt{1},\ldots,\xt{K-1}$ are omitted and we write $\theta_k=\frac{\theta}{K}+\frac{k-1}{K}$):

\begin{align*}
&\quad f_{\Xt{1},\dots,\Xt{K-1}|T(\Xt{1},\dots,\Xt{K})=t}(\xt{1},\dots,\xt{K-1}) \\ 
&\propto f_{\Xt{1},\dots,\Xt{K-1},T(\Xt{1},\dots,\Xt{K})}(\xt{1},\dots,\xt{K-1},t) \\
&= f_{\Xt{1},\dots,\Xt{K}}\left(\xt{1},\dots,\xt{K-1},\frac{t^K}{\prod_{k=1}^{K-1}\xt{k}}\right)\left|\frac{\partial}{\partial t}\frac{t^K}{\prod_{k=1}^{K-1}\xt{k}}\right| \\
&=\left(\prod_{k=1}^{K-1}f_{\Xt{k}}\left(\xt{k}\right)\right)f_{\Xt{K}}\left(\frac{t^K}{\prod_{k=1}^{K-1}\xt{k}}\right)\left|\frac{Kt^{K-1}}{\prod_{k=1}^{K-1}\xt{k}}\right| \\
&\propto \left(\prod_{k=1}^{K-1}\left(\xt{k}\right)^{\theta_k-1}\exp\left(-\frac{\beta}{K}\xt{k}\right)\right) \left(\frac{t^K}{\prod_{k=1}^{K-1}\xt{k}}\right)^{\theta_K-1}\exp\left(-\frac{\beta}{K}\frac{t^K}{\prod_{k=1}^{K-1}\xt{k}}\right)\frac{1}{\prod_{k=1}^{K-1}\xt{k}} \\
&\propto \left(\prod_{k=1}^{K-1}\left(\xt{k}\right)^{\frac{k-K}{K}-1}\right)\exp\left(-\frac{\beta}{K}\left(\sum_{k=1}^{K-1}\xt{k}+\frac{t^K}{\prod_{k=1}^{K-1}\xt{k}}\right)\right),
\end{align*}
where in the last step we used that $\theta_k-\theta_K=\frac{k-K}{K}$.

To generate $(\Xt{1},\dots,\Xt{K})$, first sample from $(\Xt{1},\dots,\Xt{K-1})|T(\Xt{1},\dots,\Xt{K})=t$ with numerical sampling methods. In this example, MCMC methods work well though the choice of proposal distribution should consider $K$ and $\beta$. Then, compute $\Xt{K}=\frac{t^K}{\prod_{k=1}^{K-1}\Xt{k}}$. 
\end{enumerate}
\end{proof}

\subsection{Families with support controlled by an unknown parameter}

\subsubsection{Scaled beta distribution}
\label{subsec:scaled-beta-proofs}

We consider the family of distributions obtained by scaling a $\text{Beta}(\alpha,1)$ distribution (with $\alpha$ fixed) by an unknown scale parameter $\theta>0$.  In the special case that $\alpha=1$, this corresponds to the $\text{Unif}(0,\theta)$ family presented in Example~\ref{ex:scaled-uniform}.

\begin{exmp}[Thinning $\theta\cdot\text{Beta}(\alpha,1)$ with $\alpha$ known] \label{ex:scaled-beta}
We start with  $\Xt{k} \overset{iid}{\sim} \theta\cdot\text{Beta}\left(\frac{\alpha}{K},1\right)$ for $k=1,\ldots,K$, and note that $T(\Xt{1},\dots,\Xt{K}) = \max(\Xt{1},\dots,\Xt{K})$ is sufficient for $\theta$. 
 Furthermore, $\max(\Xt{1},\dots,\Xt{K}) \sim \theta\cdot\text{Beta}(\alpha, 1)$. Thus, we define $G_t$ to be the conditional distribution of 
 $(\Xt{1},\dots,\Xt{K})$ given $\max(\Xt{1},\dots,\Xt{K}) = t$. Then, 
 by Theorem~\ref{thm:generalized-thinning}, 
 we can thin  $X \sim \theta\cdot\text{Beta}(\alpha,1)$ by 
 sampling from  $G_X$.

To sample from this conditional distribution, we first   
draw $\mathbf{C}\sim\text{Categorical}_K\left(1/K, \ldots, 1/K\right)$. Then, $\Xt{k} = C_kX+(1-C_k)Z_k$ where $Z_k \overset{iid}{\sim} X\cdot\text{Beta}\left(\frac{\alpha}{K},1\right)$.
\end{exmp}

\begin{proof}[Proof of Example~\ref{ex:scaled-beta}]
We must prove the following three claims:
\begin{enumerate}
    \item \emph{If for $k=1,\dots,K$, $\Xt{k} \overset{iid}{\sim} \theta \cdot\text{Beta}\left(\frac{\alpha}{K},1\right)$, then $\max(\Xt{1}, \dots,\Xt{K}) \sim \theta \cdot\text{Beta}\left(\alpha,1\right)$.}  
    
First, note that $\frac{1}{\theta}\Xt{k} \overset{iid}{\sim} \text{Beta}\left(\frac{\alpha}{K},1\right)$. The distribution of $\max(\Xt{1},\dots,\Xt{K})$ can be derived using the CDF method as follows:
\begin{align*}
    P(\max(\Xt{1},\dots,\Xt{K})\le z) &= P(\Xt{1} \le z,\dots, \Xt{K} \le z) \\
    &= \prod_{k=1}^K P(\Xt{k} \le z) \\
    &= \prod_{k=1}^K P\left(\frac{1}{\theta}\Xt{k} \le \frac{z}{\theta}\right) \\
    &= \prod_{k=1}^K \left(\frac{z}{\theta}\right)^{\frac{\alpha}{K}} \\
    &= \left(\frac{z}{\theta}\right)^\alpha,
\end{align*}
where we have used that $P(\text{Beta}(\alpha,1)\le x)=x^\alpha$ for $x\in(0,1)$. 
The above implies that $\max(\Xt{1},\dots,\Xt{K})\sim\theta \cdot\text{Beta}\left(\alpha,1\right)$ as required.
    \item \emph{A sufficient statistic for $\theta$ based on $\Xt{1},\dots,\Xt{K}$ is $\max(\Xt{1},\dots,\Xt{K})$.} 

Using the independence of $\Xt{k}$, the joint distribution can be written as

\begin{align*} 
&\quad f_{\Xt{1},\dots,\Xt{K}}(\xt{1},\dots,\xt{K}) \\
&= \prod_{k=1}^K f_{\Xt{k}}(\xt{k}) \\
&\propto \prod_{k=1}^K\left(\xt{k}\right)^{\frac{\alpha}{K}-1}I\{0<\xt{k}<\theta\} \\
&= \left(\prod_{k=1}^K\xt{k}\right)^{\frac{\alpha}{K}-1}I\{\min(\xt{1},\dots,\xt{K})>0\}I\{\max(\xt{1},\dots,\xt{K})<\theta\}.
\end{align*}

By the factorization theorem, we conclude that $\max(\Xt{1},\dots,\Xt{K})$ is a sufficient statistic for $\theta$ as required.

    \item \emph{We can sample from the conditional distribution $(\Xt{1},\dots,\Xt{K})|\max(\Xt{1},\dots,\Xt{K})=t$ by taking $\Xt{k}=C_kt+(1-C_k)Z_k$ where $\mathbf{C} \sim \text{Categorical}_K(1/K,\dots,1/K)$ and $Z_k\sim t\cdot\text{Beta}(\frac{\alpha}{K},1)$.} 

Without loss of generality, consider $\Xt{k}$. Given that $\Xt{1},\dots,\Xt{K}$ are identically distributed, $P(\Xt{k}=t)=\frac{1}{K}$. Hence, in the first stage, we can draw one sample, $\mathbf{C}\sim\text{Categorical}_K(1/K,\dots,1/K)$ to determine if $\Xt{k}$ is the maximum. If $\Xt{k}$ is not the maximum then we know that $\Xt{k} \le t$. We can compute the density of $Z_k\overset{D}{=}(\Xt{k}|\Xt{k}\le t)$ as follows,

$$
f_{\Xt{k}|\Xt{k}\le t}(\xt{k}) = \frac{f_{\Xt{k}}(\xt{k})}{P(\Xt{k}\le t)} =\frac{\frac{1}{\theta^\frac{\alpha}{K} B\left(\frac{\alpha}{K},1\right)} \left(\xt{k}\right)^{\frac{\alpha}{K}-1}}{\left(\frac{t}{\theta}\right)^{\frac{\alpha}{K}}} = \frac{1}{t^\frac{\alpha}{K} B\left(\frac{\alpha}{K},1\right)} \left(\xt{k}\right)^{\frac{\alpha}{K}-1}.
$$

The above implies that $Z_k \sim t\cdot\text{Beta}\left(\frac{\alpha}{K},1\right)$ as required.
    
\end{enumerate}
The result then follows from Theorem~\ref{thm:generalized-thinning}.
\end{proof}

\subsubsection{Shifted exponential distribution}
\label{subsec:shifted-exponential-proofs}

We consider $X\sim \mathrm{SExp}(\theta,\lambda)$, which is the location family generated by shifting an exponential random variable by an amount $\theta$.  It has density
$$
p_{\theta,\lambda}(x)=\lambda e^{-\lambda (x-\theta)}1\{x\ge \theta\}.
$$

\begin{exmp}[Thinning a $\mathrm{SExp}(\theta,\lambda)$ random variable with known $\lambda$] \label{ex:shifted-exponential}
We begin with $\Xt{k} \overset{iid}{\sim} \mathrm{SExp}(\theta,\lambda/K)$ for $k=1,\ldots,K$, and note that $T(\Xt{1},\dots,\Xt{K})=\min(\Xt{1},\dots,\Xt{K})$ is sufficient for $\theta$.  Furthermore,  $\min(\Xt{1},\dots,\Xt{K})\sim \mathrm{SExp}(\theta,\lambda)$. We define $G_t$ to be the conditional distribution of $(\Xt{1},\ldots,\Xt{K})$ given $  \min(\Xt{1},\dots,\Xt{K})=t$.  Then, by Theorem~\ref{thm:generalized-thinning},   we can thin $X\sim \mathrm{SExp}(\theta,\lambda)$  by sampling from $G_X$.

To sample from $G_X$, we first draw  $\mathbf{C}\sim\text{Categorical}_K\left(1/K,\ldots,1/K \right)$. We then take 
$\Xt{k}=X+(1-C_k)Z_k$,
where  $Z_k\overset{iid}{\sim} \mathrm{Exp}(\lambda/K)$. 

\end{exmp}

\begin{proof}[Proof of Example~\ref{ex:shifted-exponential}]
We must prove the following three claims:
\begin{enumerate}
    \item \emph{If for $k=1,\dots,K$, $\Xt{k} \overset{iid}{\sim} \text{SExp}\left(\theta,\lambda/K\right)$, then $\min(\Xt{1}, \dots,\Xt{K}) \sim \text{SExp}(\theta,\lambda)$.}  
    
First, note that $\Xt{k} - \theta \overset{iid}{\sim} \text{Exp}\left(\lambda/K\right)$. The distribution of $\min(\Xt{1},\dots,\Xt{K})$ can be derived using the CDF method as follows,
\begin{align*}
    P(\min(\Xt{1},\dots,\Xt{K})\ge z) &= P(\Xt{1} \ge z,\dots, \Xt{K} \ge z) \\
    &= \prod_{k=1}^K P(\Xt{k} \ge z) \\
    &= \prod_{k=1}^K P\left(\Xt{k}-\theta \ge z-\theta\right) \\
    &= \prod_{k=1}^K \exp\left(-\frac{\lambda}{K}(z-\theta)\right) \\
    &= \exp\left(-\lambda(z-\theta)\right).
\end{align*}

The above implies that $\min(\Xt{1},\dots,\Xt{K})\sim\text{SExp}(\theta,\lambda)$ as required.
    \item \emph{A sufficient statistic for $\theta$ based on $\Xt{1},\dots,\Xt{K}$ is $\min(\Xt{1},\dots,\Xt{K})$.} 

Using the independence of $\Xt{k}$, the joint distribution can be written as
\begin{align*} 
&\quad f_{\Xt{1},\dots,\Xt{K}}(\xt{1},\dots,\xt{K}) \\
&= \prod_{k=1}^K f_{\Xt{k}}(\xt{k}) \\
&\propto \prod_{k=1}^K \exp\left(-\frac{\lambda}{K}(\xt{k}-\theta)\right)I\{\xt{k}>\theta\} \\
&\propto \exp\left(-\frac{\lambda}{K}\sum_{k=1}^K\xt{k}\right)I\{\min(\xt{1},\dots,\xt{K})>\theta\}. 
\end{align*}

Given that the joint distribution can be written such that $\theta$ only interacts with the data through the $I\{\min(\xt{1},\dots,\xt{K})>\theta\}$ term, we conclude that $\min(\Xt{1},\dots,\Xt{K})$ is a sufficient statistic for $\theta$ as required.

    \item \emph{We can sample from the conditional distribution $\left(\Xt{1},\dots,\Xt{K}\right)|\min(\Xt{1},\dots,\Xt{K})=t$ by taking $\Xt{k}=t+(1-C_k)Z_k$ where $\mathbf{C} \sim \text{Categorical}_K(1/K,\dots,1/K)$ and $Z_k\sim\text{Exp}(\lambda/K)$.} 

Without loss of generality, consider $\Xt{k}$. Given that $\Xt{1},\dots,\Xt{K}$ are identically distributed, $P(\Xt{k}=X)=\frac{1}{K}$. Hence, in the first stage, we can draw one sample, $\mathbf{C}\sim\text{Categorical}_K(1/K,\dots,1/K)$ to determine if $\Xt{k}$ is the minimum. Otherwise, we require that $ \Xt{k} \ge t$. We know that the density of $Z_k \overset{D}{=}(\Xt{k}|\Xt{k}\ge t) \overset{D}{=} \Xt{k}$ by the memoryless property of the exponential distribution. Thus, $Z_k \sim \text{Exp}\left(\lambda/K\right)$ as required.
    
\end{enumerate}
The result then follows from Theorem~\ref{thm:generalized-thinning}.

\end{proof}

\section{Additional example of indirect thinning}
\label{appendix:indirect}

We consider an example of indirect thinning in which $S(\cdot)$ is neither invertible, nor scalar-valued. Specifically, let $\mathbf{X}=(X_1,\ldots,X_n)$ represent a sample of $n$ independent and identically distributed normal observations with unknown mean $\theta_1$ and variance $\theta_2$. Using ideas from Theorem~\ref{thm:natural-exponential-families} and Proposition~\ref{prop:exp-family}, we thin $\mathbf{X}$ through the sample mean and sample variance. 

\begin{exmp}[Indirect thinning of  $N_n(\theta_1\mathbf{1}_n, \theta_2\mathbf{I}_n)$  through the sample mean and sample variance] \label{ex:normal} 
Suppose $\mathbf{X}\sim N_n(\theta_1\mathbf{1}_n, \theta_2\mathbf{I}_n)$. Then $S(\mathbf{X})$ is sufficient for $\boldsymbol\theta=(\theta_1,\theta_2)$, where
$$
S(\mathbf{x}) =  \left( \frac{1}{n} \sum_{i=1}^n x_i, \frac{1}{n-1}\sum_{i=1}^{n}\left(x_i -  \frac{1}{n} \sum_{i'=1}^n x_{i'}\right)^2 \right).
$$ 
We will indirectly thin $\mathbf{X}$ through $S(\cdot)$ into $K=n$ univariate normals.
To do so, we start with  $\Xt{k} \overset{iid}{\sim} N\left(\theta_1,\theta_2\right)$ for $k=1,\ldots,K$. A sufficient statistic for $\boldsymbol\theta$ based on $(\Xt{1},\ldots,\Xt{K})$ is $T(\Xt{1},\ldots,\Xt{K})$,  where
$T(\xt{1},\ldots,\xt{K})=S((\xt{1},\ldots,\xt{K})^\top)$,
i.e., we concatenate the $K$ entries into a vector and apply $S(\cdot)$.  
Furthermore, $T(\Xt{1},\ldots,\Xt{K})$ has the same distribution as $S(\mathbf{X})$, since  $(\Xt{1},\ldots,\Xt{K})^\top$ and $\mathbf{X}$ have the same distribution. 
This establishes that we can indirectly thin $\mathbf{X}$ through $S(\cdot)$ by $T(\cdot)$. 

By Theorem~\ref{thm:generalized-thinning}, we define $G_{\mathbf{t}}$ to be the conditional distribution of  $(\Xt{1},\ldots,\Xt{K})$ given $T(\Xt{1},\ldots,\Xt{K})=\mathbf{t}$, i.e. it is the distribution of a $N_K(\theta_1\mathbf{1}_K, \theta_2\mathbf{I}_K)$ random vector conditional on its sample mean and sample variance equalling $\mathbf{t}=(t_1,t_2)$.  
This conditional distribution is uniform over the set of points in $\mathbb{R}^K$ with sample mean $t_1$ and sample variance $t_2$.  To see this, note that $G_{\mathbf{t}}$ cannot depend on $\boldsymbol\theta$ (by sufficiency), so we can take $\boldsymbol\theta=(t_1,t_2)$ and equivalently describe $G_{\mathbf{t}}$ as the distribution of a $N_K(t_1\mathbf{1}_K, t_2\mathbf{I}_K)$ random vector conditional on its sample mean and sample variance equalling $\mathbf{t}=(t_1,t_2)$.  Such a distribution has constant density on a sphere centered at $t_1\mathbf{1}_K$.  Thus, the conditional distribution $G_{\mathbf{t}}$ is uniform over the set of points in $\mathbb{R}^K$ with sample mean $t_1$ and sample variance $t_2$. Finally, we obtain $(\Xt{1},\ldots,\Xt{K})$ by sampling from $G_{\mathbf{X}}$.
\end{exmp}
In effect, Example~\ref{ex:normal} shows that given a realization of $\mathbf{X}\sim N_n(\theta_1\mathbf{1}_n, \theta_2\mathbf{I}_n)$, we can generate a new random vector $(\Xt{1},\ldots,\Xt{n})^\top$ with the identical distribution,  and with the same sample mean and sample variance, without knowledge of the true mean or true variance. 
This might have applications in cases where the true values of the observations cannot be shared. 
Furthermore, the ideas in  Example~\ref{ex:normal} can be applied in settings where only the sufficient statistics of a realization of $\mathbf{X}\sim N_n(\theta_1\mathbf{1}_n, \theta_2\mathbf{I}_n)$ are available, and we wish to generate a ``plausible" sample that could have led to those sufficient statistics. 

\section{Numerical experiments}
\label{app:experiments}

In this section, we illustrate some of the examples from Sections~\ref{sec:natural-exp-fam}, \ref{sec:general-exp}, and \ref{sec:outside-exp-fam} through numerical simulations.
Specifically, we thin a Gamma($\alpha, \theta$) distribution using the three different approaches described in Examples~\ref{ex:dtgamma}, \ref{ex:scaled-normal}, and \ref{ex:weibull}, a Beta($\theta, \beta$) into two non-identical beta random variables as described in Example~\ref{ex:beta}, and a Unif($0, \theta$) into scaled betas as described in Example~\ref{ex:scaled-uniform}.  We take $K=2$ throughout for ease of presentation.

\begin{figure}[!h]
\begin{center}
\includegraphics[width=\textwidth]{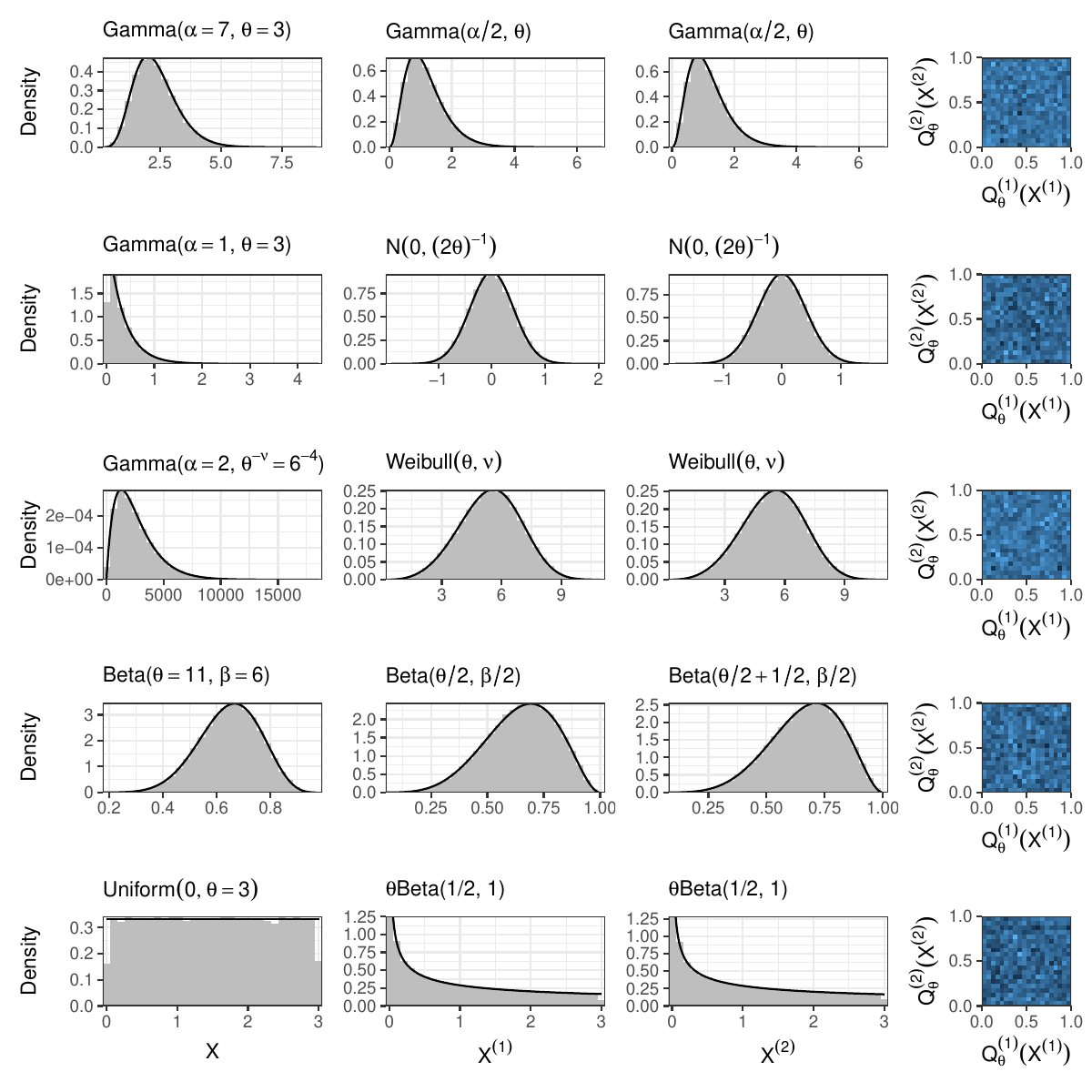}
\end{center}
\caption{Numerical examples of data thinning. The left-hand column displays a sample from $P_\theta$, which we wish to thin.  The center-left and center-right columns display the empirical distributions of $\Xt{1}$ and $\Xt{2}$ that result from thinning, overlaid with  the theoretical distributions $\Qt{1}_\theta$ and $\Qt{2}_\theta$. The right-hand column displays the empirical joint distribution of $(\Qt{1}_\theta(\Xt{1}), \Qt{2}_\theta(\Xt{2}))$, providing visual evidence that they are independent. With a slight abuse of notation, $\Qt{1}_\theta(\cdot)$ and $\Qt{2}_\theta(\cdot)$ represent the CDFs of their respective distributions.}
\label{fig:numericalexamples}
\end{figure}

In Figure \ref{fig:numericalexamples}, each row corresponds to one of the examples mentioned above. In the left-hand column,  we display the empirical density of $B=100,000$ realizations, $x_b$ (for $b=1,\ldots,B$), of the $X\sim P_\theta$ that we wish to thin, overlaid with the true density of $P_\theta$. In the center-left and center-right columns, we display $B$ realizations of $\Xt{1}$ and $\Xt{2}$ respectively, where each realization $(\xt{1}_b,\xt{2}_b)$ is obtained by sampling from $G_{x_b}$, the conditional distribution of $(\Xt{1},\Xt{2})$ given $T(\Xt{1}, \Xt{2})=x_b$. (This sampling is done without knowledge of $\theta$.) We overlay the densities of the marginals $\Qt{1}_\theta$ and $\Qt{2}_\theta$. The right-hand column displays 
the empirical joint distribution of $(\Qt{1}_\theta(\Xt{1}), \Qt{2}_\theta(\Xt{2}))$.  

In each case, our empirical findings corroborate our theoretical results: We see that the empirical distribution of $X \sim P_\theta$ agrees with its theoretical density (left-hand column); that the empirical distributions of $\Xt{1}$ and $\Xt{2}$ sampled from $G_X$ coincide with $\Qt{1}_\theta$ and $\Qt{2}_\theta$ (even though the empirical distributions were obtained without knowledge of $\theta$; center-left and center-right columns); and that the joint distribution of $\Qt{1}_\theta(\Xt{1})$ and $\Qt{2}_\theta(\Xt{2})$ resembles the independence copula (right-hand column). 

\section{Changepoint detection simulations} \label{app:changepoint}
First, we generate data with a common variance, specifically $X_1,\ldots,X_{2000}\overset{\text{iid}}{\sim} N(0,1)$, and apply the three approaches to detecting and testing for a change in variance that were described in Section \ref{sec:changepoint}.  
We repeat this process 1000 times, and display the type 1 error rate in Figure \ref{fig:simulation}. The naive approach does not control the type 1 error rate, while the order-preserved sample splitting and generalized data thinning approaches do. 

Figure \ref{fig:null} displays the estimated changepoints, as well as those for which we rejected the null hypothesis of no change in variance at the Bonferroni corrected threshold, for a single realization of the simulated data. The naive approach results in a number of false positives, while the order-preserved sample splitting and generalized data thinning approaches do not. 

\begin{figure}[!h]
\begin{center}
\includegraphics[width=\textwidth]{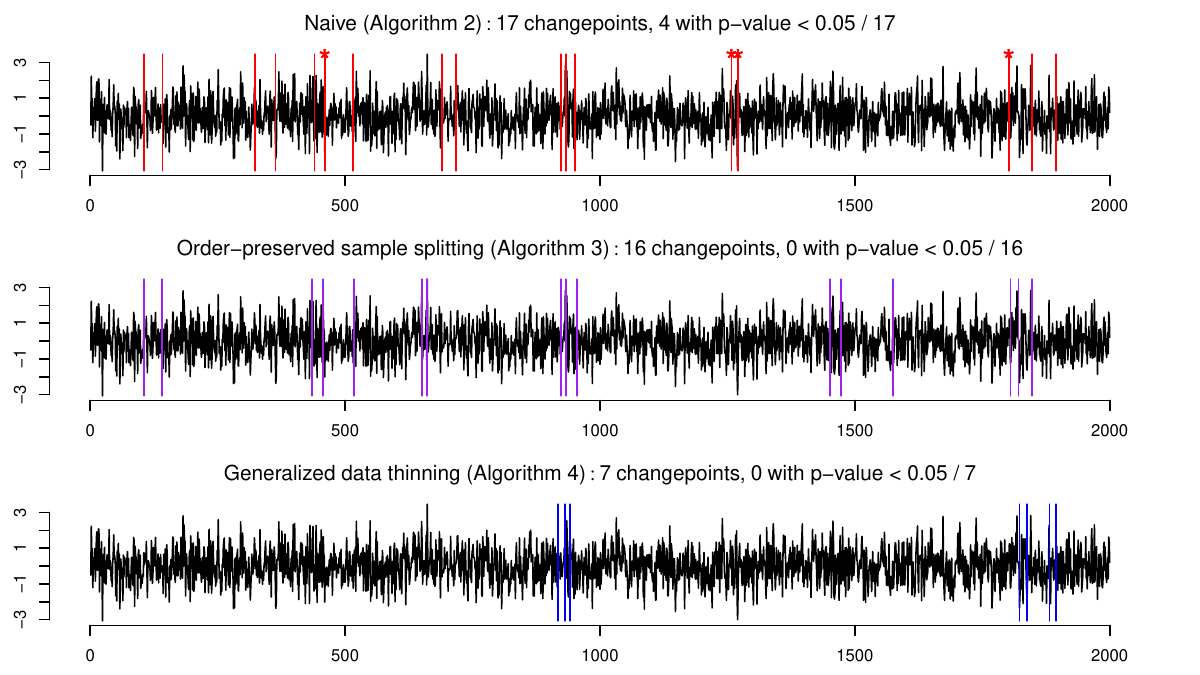}
\end{center}
\caption{Results for the simulation study with no changepoints. \emph{Top row:} Simulated data with constant variance, as well as results of the naive approach: Red lines indicate changepoints estimated using all of the data, and red asterisks indicate the estimated changepoints for which the p-value computed using all of the data was below 0.05  divided by the number of detected changepoints. The method (falsely) rejects the null hypothesis for 4 out of 17 estimated changepoints. 
\emph{Middle row:} Same as the first row, but for the order-preserved sample splitting approach: Purple lines indicate changepoints estimated using the odd observations, and purple asterisks indicate those with p-values below 0.05  divided by the number of detected changepoints, using only the even observations for testing. This method (correctly) rejects none of the 16 estimated changepoints.
\emph{Bottom row:} Same as the first row, but for the generalized
data thinning approach: Blue lines indicate changepoints estimated using the training set, and blue asterisks indicate those with test set p-values below 0.05  divided by the number of detected changepoints. This method (correctly) rejects none of the 7 estimated changepoints.}
\label{fig:null}
\end{figure}

Next, we generated data with two true changepoints: for $i = 1, \dots, 500$, $X_i \overset{\text{iid}}{\sim} N(0,4)$; for $i = 501, \dots, 1500$, $X_i \overset{\text{iid}}{\sim} N(0,25)$; and for $i = 1501, \dots, 2000$, $X_i \overset{\text{iid}}{\sim} N(0,1)$. We again apply the three approaches to detecting and testing for a change in variance that were described in Section \ref{sec:changepoint}, and display the results in Figure \ref{fig:alt}. In this setting, all three approaches reject the null hypothesis of no change in variance at two timepoints, which are located exactly at the two true changepoints. 

\begin{figure}[!h]
\begin{center}
\includegraphics[width=\textwidth]{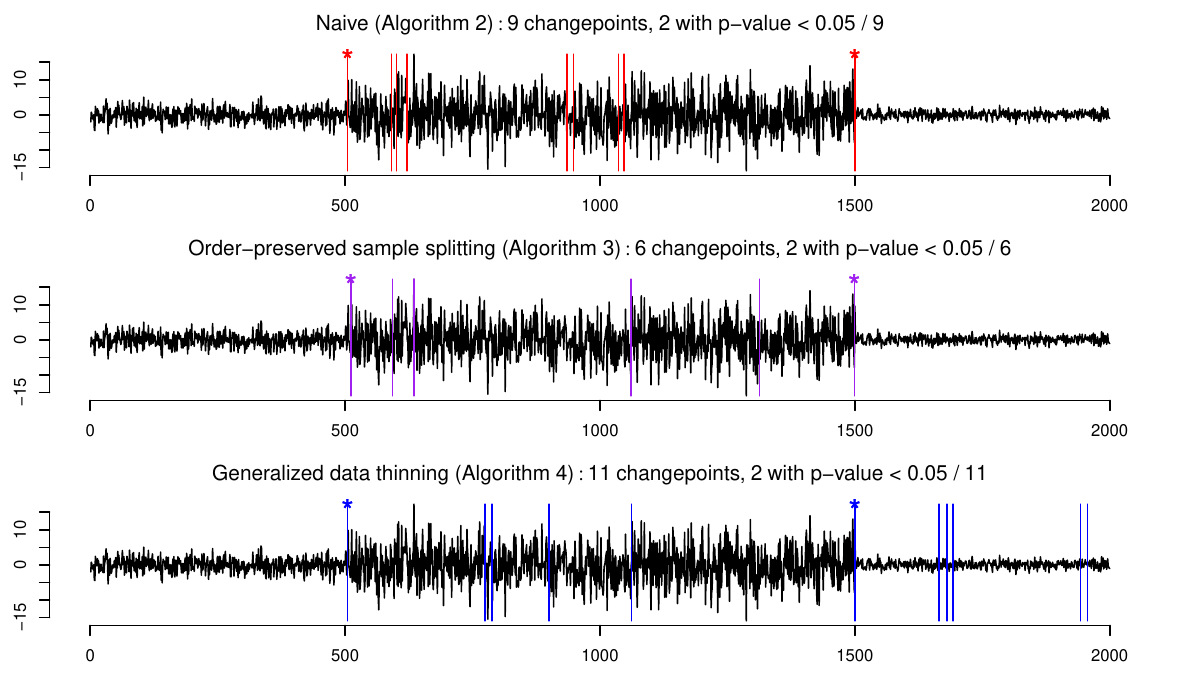}
\end{center}
\caption{Results for the simulation study with changes in variance at timepoints 501 and 1501. \emph{Top row:} Simulated data with two changepoints, as well as results of the naive approach: Red lines indicate changepoints estimated using all of the data, and red asterisks indicate the estimated changepoints for which the p-value computed using all of the data was below 0.05 divided by the number of detected changepoints. This approach leads to rejections only at the two true changepoints.
\emph{Middle row:} Same as the first row, but for the order-preserved sample splitting approach: Purple lines indicate changepoints estimated using the odd observations, and purple asterisks indicate those with p-values below 0.05 divided by the number of detected changepoints, using only the even observations for testing. This approach also leads to rejections only at the two true changepoints. \emph{Bottom row:} Same as the first row, but for the generalized
data thinning approach: Blue lines indicate changepoints estimated using the training set, and blue asterisks indicate those with test set p-values below 0.05 divided by the number of detected changepoints. This approach also leads to rejections only at the two true changepoints.}
\label{fig:alt}
\end{figure}

\section{Relaxing the independence assumption} \label{app:fission}

We now consider how the generalized data thinning recipe changes if we relax the independence requirement for $\Xt{1}$ and $\Xt{2}$.

\begin{algorithm}[Finding distributions that can be decomposed into non-independent components]
\label{alg:recipe-fission}
\textcolor{white}{.}
\begin{enumerate}
    \item Choose a family of distributions $\mathcal Q=\{Q_\theta:\theta\in\Omega\}$ over $(\Xt{1},\Xt{2})$, 
     where $\Xt{1}$ and $\Xt{2}$ are not necessarily independent.
    \item Let $(\Xt{1},\Xt{2})\sim Q_\theta$, and let  $T(\Xt{1},\Xt{2})$ denote a sufficient statistic for $\theta$.
    \item Let $P_\theta$ denote the distribution of $T(\Xt{1},\Xt{2})$.  
\end{enumerate}
\end{algorithm}
Then, given $X \sim P_\theta$, we can generate $(\Xt{1},\Xt{2})$ by sampling from $G_X$, where $G_t$ is defined as the conditional distribution
    $$
    (\Xt{1},\Xt{2})|T(\Xt{1},\Xt{2})=t.
    $$

By sufficiency, the sampling mechanism $G_t$  can be performed without knowledge of $\theta$.
The key point here is that the main ideas in this paper apply even if $\Xt{1}$ and $\Xt{2}$ are dependent; however, we focused on independence in this paper to facilitate downstream application of the decompositions that we obtain.

\end{document}